\documentclass[journal]{IEEEtran}

\usepackage{mathtools}
\usepackage{amsmath}
\usepackage{color,soul}			
\usepackage{graphicx}	
\usepackage{tikz}
\usepackage{amssymb}
\usepackage{amsthm}  
\usepackage{subfig,balance,ulem}
\usepackage{cite}

\usepackage{tikz}
\usetikzlibrary{shapes,arrows}

\newtheorem{theorem}{Theorem}
\newtheorem{corollary}{Corollary}
\newtheorem{lemma}{Lemma}
\newtheorem{definition}{Definition}
\newtheorem{example}{Example}
\newtheorem{proposition}{Proposition}

\newcommand\LG{\mathord{\mathcal{L}}}
\newcommand\LM{\mathord{\mathcal{L}}_m}
\newcommand\LANG{\mathord{\mathcal{L}}}

\begin{document}

\title{Secure Recovery Procedure for Manufacturing Systems using Synchronizing Automata and Supervisory Control Theory}

\author{L. V. R. Alves,          
        P. N. Pena 

\thanks{Lucas V. R. Alves is with
              Technical College and Graduate Program in Electrical Engineering,
              Universidade Federal de Minas Gerais 
              (e-mail: lucasvra@ufmg.br) 
}%
\thanks{
           Patr\'icia N. Pena is with
              Department of Electronics Engineering, 
              Universidade Federal de  Minas Gerais 
              (e-mail: ppena@ufmg.br)     
}%
\thanks{This work has been supported by the Brazilian agency CAPES, CNPq and Fapemig.}
\thanks{This work has been submitted to the IEEE for possible publication. 
Copyright may be transferred without notice, after which this version may no longer be accessible.}}%

\maketitle

\begin{abstract}
Manufacturing systems may be subject to external attacks and failures, so it is important to deal with the recovery of the system after these situations. This paper deals with the problem of recovering a manufacturing system, modeled as a Discrete Event System (DES) using the Supervisory Control Theory (SCT), when the control structure, called supervisor, desynchronizes from the physical plant.
The desynchronization may be seen as plant and supervisor being in uncorresponding states. The recovery of the system may be attained if there is a word, the synchronizing word, that regardless the state of each one of them, brings the system and supervisor back to a known state. The concepts of synchronizing automata are used to do so. In this paper we show under what conditions a set of synchronizing plants and specifications leads to a synchronizing supervisor obtained by the Supervisory Control Theory. The problem is extended to cope with multiple supervisors, proposing a local recovery when possible. We also present a simple way to model problems, composed of machines and buffers, as synchronizing automata such that it is always possible do restore synchronization between the control (supervisor) and the plant.

\end{abstract}

\def\abstractname{Note to Practitioners}
\begin{abstract}
Given the unpredictability of faults and malicious attacks occurring in industrial systems, recovery strategies are crucial for a harmonic operation of the plant. The possibility of leading the system to a known state, recovering control, is of extreme importance to the safety of industrial processes. The method proposed in this paper uses well known concepts of Supervisory Control Theory (SCT) of Discrete Event Systems (DES), introducing the recovery process (using recovery events) in the modeling phase such that it is possible to isolate and fix only the part of the control system subject to the fault.
The result of the proposed approach allows the implementation of such control system with the recovery procedure directly in the Programmable Logic Controllers (PLCs).

\end{abstract}


\IEEEpeerreviewmaketitle

\section{Introduction}\label{intro}

Fault recovery is an essential part of a modern manufacturing system. Most of the data in \textit{Smart Plants} is accessed over real-time communication networks, so, in addition to worrying about sensor and actuator failures, we also need to take into account malicious attacks to the system. In computational systems, such problems can be solved restarting the software, but in industries, because of safety and reliability constraints, this restart cannot be naive \cite{Abad2016}.

In the Supervisory Control Theory (SCT), the supervisor restricts the dynamics of the system inhibiting the execution of controllable events in order to guarantee a safe operation of the system. As shown in Fig.~\ref{fig:attack}, the supervisor estimates the current state of the plant by observing the occurrence of events, however this observation is susceptible to problems originated by malicious attacks and communication problems, leading the system to a situation where the physical state of the plant does not correspond to the state estimated by the supervisor. In some situations, the observations made by the supervisor can be corrupted as the list of allowed events sent to the plant.

Starting in the decade of 2000, the increase in the exchange of information in digital environment increases the concern with the security of computational systems \cite{Saboori2007}. Every system with communication among its agents, as between plant and supervisor, is susceptible to attacks.

Such attacks are becoming more sophisticated having as their main objectives to steal information, extortion and sabotage \cite{Tankard201116, Beuhring2014}. 
APT - \textit{Advanced Persistent Threat} are pieces of software developed to attack specific targets \cite{Virvilis2013} and stay hidden in these systems for long periods of time.

Another cause of problems in manufacturing systems are the failures in sensors, actuators and communication systems. Most of the information that travels in intelligent manufacturing systems is accessed by real time communication networks \cite{Christofides2007} and this information may be corrupted or lost.

The problems of recovery of Discrete Event Systems can be divided into three sub-problems \cite{Loborg1994}:
\begin{enumerate}
    \item Detection: Consists in detecting discrepancies between the state of the system and the specifications/supervisor \cite{Carvalho2018}.
    \item Diagnostic: Consists in detecting the fault that generated the discrepancy. In DES, this problem may be handled using techniques of diagnosability using automata models of Discrete Event Systems \cite{Lafortune2018}.
    \item Recovery: After eliminating the cause of the fault, the malicious agent or faulty parts, the recovery may be about changing the state of the system and supervisor to be consistent.
\end{enumerate}

Shu~\cite{Shu2014} deals with the recovery of manufacturing systems firing \textit{recovery events} when an event sequence leads the system to a faulty mode. These \textit{recovery events} cannot be disabled by the supervisor and they are used by the supervisor in order to recover the system. On the other hand, Andersson and coauthors~\cite{Andersson2009, Andersson2010, Andersson2011}, Bergagard and coauthors~\cite{Bergagard2013, Bergagard2015} present a method to restart manufacturing systems, modeled using operations and \textit{coordination of operations} (COP), after unforeseen errors using the notion of \textit{restart states}. In this context, the restart process act by resynchronizing the physical state of a plant with the state of COP.

In this paper, we propose the use of the theoretical development in \textit{Synchronizing Automata} to deal with the problem when the active state of the plant does not match the active state of the supervisor, after the system suffers an attack of a malicious agent or after a fault. In this sense, we consider that the system loses synchronization when the active control state does not correspond to the active physical state.

\tikzstyle{block1} = [draw, fill=blue!5, rectangle, minimum height=3em, minimum width=6em]
\tikzstyle{block2} = [draw, fill=blue!5, rectangle, minimum height=3em, minimum width=3em]
\tikzstyle{pinstyle} = [pin edge={to-,thin,black}]

\begin{figure}
    \centering
    \begin{tikzpicture}[auto, node distance=2cm,>=latex]
    \node[block1] at (2, 3) (supervisor) {Supervisor};
    \node[block1] at (2, 0) (plant) {Plant};
    \node[block2] at (0, 1.5) (obs1) {Attacker};
    \node[block2] at (4, 1.5) (obs2) {Attacker};
    
   \draw [->] (plant) node[right, align=center, xshift=2cm, yshift=0.5cm] {Observations} -| (obs2);
   \draw [->] (obs2)  node[right, align=center, xshift=0cm, yshift=1.5cm] {Corrupt\\Observations} |- (supervisor);
   \draw [->] (supervisor) node[left, align=center, xshift=-2cm, yshift=-0.0cm] {Allowed\\Events} -| (obs1);
   \draw [->] (obs1) node[left, align=center, xshift=0cm, yshift=-1.0cm] {Corrupt\\Allowed Events} |- (plant);
    
    \end{tikzpicture}
    \caption{System under attack or failure}
    \label{fig:attack}
\end{figure}
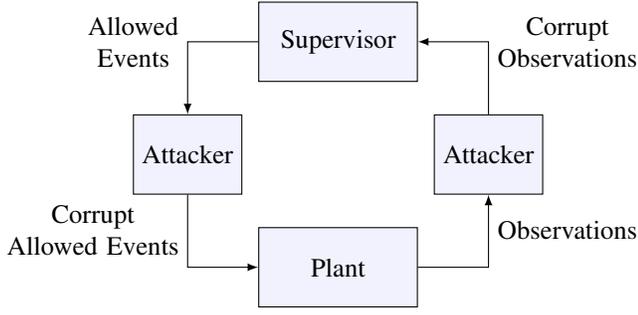

Differently from the techniques presented in \cite{Andersson2009, Andersson2010, Andersson2011, Bergagard2013, Bergagard2015}, if it is possible to model the system components as synchronizing automata then it is always possible to restart the system, and there is no need to insert additional components. We, also, present a simple method to model conventional problems of SCT as synchronizing automata inserting recovery events, similar to those presented in \cite{Shu2014}.

An automaton is said \textit{synchronizing} when there is a word, called \textit{synchronizing word}, that, when executed by the automaton, leads to the same state, regardless of the state of origin. So, two identical automata, in different states, will always evolve to the same state when a \textit{synchronizing word} is executed \cite{Volkov2008}.

The existence of a synchronizing word has applications in many fields, such as robotics, assembling, loading and packing of products \cite{Natarajan1986,Natarajan1989}. More theoretical development was presented in the context of industrial automation \cite{Eppstein1988,Goldberg1993,Chen1995}. Synchronizing automata were also applied to problems with partial observability \cite{Larsen2014} and problems modeled with Petri Nets \cite{Pocci2013,Pocci20141,Pocci20142,Pocci2016}.

The contributions of the paper are summarized. First, we present how the synchronizing automata coexist with the Supervisory Control Theory and in which cases the synchronization is maintained after the synthesis of a controllable and nonblocking supervisor. Then, we show how to turn automata that  model the plants and specifications into synchronizing automata using recovery events (alike \cite{Shu2014}) and how they can be used to resynchronize supervisor and plant.
Then, we expand the obtained results to Local Modular Supervisory Control, allowing partial recovery of the system, pointed out as a future challenge in \cite{Andersson2011}.

This paper is organized such that Section \ref{sec:pre} has the preliminaries, where we show the main concepts needed to understand the results. Section~\ref{sec:prob} states the problem this paper aims to solve. Section \ref{sec:main} presents the main results, where we present conditions under which synchronization survives the syntheses of supervisors in both in the Monolithic Supervisory Control and the Local Modular Supervisory Control. In Section \ref{sec:example}, an example is presented showing how synchronizing automata can be used in discrete event systems. The conclusions are in Section \ref{sec:conclusions}.

\section{Preliminaries}\label{sec:pre}

In this section, we summarize some fundamental concepts and results of the Supervisory Control Theory (SCT) of Ramadge and Wonham~\cite{RW89}, that are needed for the theoretical development of the paper. We, also, define some concepts and notation on the \textit{synchronizing automata}.
\subsection{Languages and Automata}
Let~$ \Sigma $ be a finite nonempty set of \textit{events}, referred to as an \textit{event set}. Behaviors of DES are modeled by finite words over~$ \Sigma $. The set of all finite words composed of events in $ \Sigma $, including the empty word~$\varepsilon$, is denoted by $ \Sigma^{*} $. A subset~$L \subseteq \Sigma^{*} $ is called a \textit{language}. The \textit{concatenation} of words $s,u \in \Sigma^*$ is written as~$su$.
A word~$s \in \Sigma^{*}$ is called a \textit{prefix} of $t \in \Sigma^{*}$, written $s \leq t$, if there exists $u \in \Sigma^{*}$ such that $su=t$. The \textit{prefix-closure} $\overline{L}$ of a language~$L \subseteq \Sigma^{*}$ is the set of all prefixes of words in~$L$, i.e., $\overline{L} = \{\, s \in \Sigma^{*} \mid s \leq t\ \mbox{for some}\ t \in L\,\}$.
A common operation over words and languages is the natural projection. Given two event sets $\Sigma$ and $\Sigma_i$, such that $\Sigma_i \subseteq \Sigma$, the natural projection $P_{\Sigma \to \Sigma_i}: \Sigma^*\rightarrow\Sigma_i^*$ is defined as:
\begin{align*}
P_{\Sigma \to \Sigma_i}(\epsilon)\phantom{s} &= \epsilon\\
P_{\Sigma \to \Sigma_i}(\sigma)  \phantom{s} &=  \begin{cases} 
   \sigma & \text{if } \sigma \in \Sigma_i \\
   \epsilon  & \text{if } \sigma \in \Sigma \setminus \Sigma_i
  \end{cases}\\
P_{\Sigma \to \Sigma_i}(s \sigma) &= P_{\Sigma \to \Sigma_i}(s)P_{\Sigma \to \Sigma_i}(\sigma) \text{, } s \in \Sigma^*, \sigma \in \Sigma.
\end{align*}

The inverse projection maps a word built from an event set $\Sigma_i$ to a language in the event set $\Sigma$ as:
$$ P^{-1}_{\Sigma \to \Sigma_i}(t) = \{ s \in \Sigma^* \,|\, P_{\Sigma \to \Sigma_i}(s) = t \}. $$

Both operations can be extended to operate over languages. For $L \subseteq \Sigma^*$:
$$ P_{\Sigma \to \Sigma_i}(L) = \{ t \in \Sigma_i^* \,|\, (\exists s \in L) \, [P_{\Sigma \to \Sigma_i}(s) = t] \}. $$

For $L \subseteq \Sigma_i^*$:
$$ P^{-1}_{\Sigma \to \Sigma_i}(L) = \{ s \in \Sigma^* \,|\, (\exists t \in P_{\Sigma \to \Sigma_i}(L)) \, [P_{\Sigma \to \Sigma_i}(s) = t] \}.$$

\begin{definition} 
A Deterministic Finite Automata (DFA) is a 5-tuple $G = (Q, \Sigma, \delta, q_0, Q_m)$, where $Q$ is a finite set of states, $\Sigma$ is an event set, $\delta : Q \times \Sigma \to Q$ is the transition function, $q_0 \in Q$ is the initial state and $Q_m \subseteq Q$ is the set of marked states. \qed
\end{definition}
The transition function can be  extended to recognize words over $\Sigma^*$ as $\delta(q, \sigma s) = q'$ with $\delta(q, \sigma) = x$ and $ \delta(x, s) = q' $. 

The execution of a word $s$ in a state $q$, $\delta(q,s)$, is denoted by the concatenation $q.s$. The same notation is used to represent this operation over sets. The notation $A.s$ denotes the set of destination states when the word $s$ is executed from the set of states $A\subseteq Q$. 

The active event function, defined by $\Gamma : Q \to 2^{\Sigma}$, is, given a state $q$, the set of events $\sigma \in \Sigma$ for which $\delta(q,\sigma)$ is defined.

The generated and marked languages are, respectively, $\mathcal{L}(G) = \{s \in \Sigma^* | q_0.s = q' \land q' \in Q \}$ and $\mathcal{L}_m(G) = \{s \in \Sigma^* | q_0.s = q' \land q' \in Q_m \}$. Another language is defined to include words  starting in any state $q$ of $G$ as $\mathcal{L}_G(q) = \{s \in \Sigma^* | q.s = q' \land q, q' \in Q \}$ such that $\mathcal{L}_G(q_0) = \mathcal{L}(G)$. An automaton is said to be nonblocking if $\overline{\mathcal{L}_m(G)} = \mathcal{L}(G)$.
\begin{definition}
Let $G_1 = (Q_1, \Sigma_1, \delta_1, q_{01}, Q_{m1})$ and $G_2 = (Q_2, \Sigma_2, \delta_2, q_{02}, Q_{m2})$ be two automata. The parallel composition of $G_1$ and $G_2$, denoted by $G_{12}=G_1||G_2$ is:
$$G_{12} = (Q_1 \times Q_2, \Sigma_1 \cup \Sigma_2, \delta_{12}, (q_{01},q_{02}), Q_{m1} \times Q_{m2})$$
where 
\footnotesize
\begin{equation*}
\delta((q_1, q_2), \sigma) = \begin{cases}
(\delta_1(q_1,\sigma),\delta_2(q_2,\sigma)), & \text{if $\sigma \in \Gamma_1(q_1) \cap \Gamma_2(q_2)$} \\
(\delta_1(q_1,\sigma),q_2), & \text{if $\sigma \in \Gamma_1(q_1) \backslash \Sigma_2$} \\
(q_1,\delta_2(q_2,\sigma)), & \text{if $\sigma \in \Gamma_2(q_2) \backslash \Sigma_1$} \\
\mbox{undefined}, & \text{otherwise}.
\end{cases}
\end{equation*}
\normalsize

Also, let $P_{\Sigma_1 \cup \Sigma_2\to\Sigma_1}:(\Sigma_1 \cup \Sigma_2)^*\to\Sigma_1^*$  and $P_{\Sigma_1 \cup \Sigma_2\to\Sigma_2}:(\Sigma_1 \cup \Sigma_2)^*\to\Sigma_2^*$ be natural projections:

\begin{align}
\mathcal{L}(G_{12}) &= P^{-1}_{\Sigma_1 \cup \Sigma_2\to\Sigma_1}(\mathcal{L}(G_{1})) \cap P^{-1}_{\Sigma_1 \cup \Sigma_2\to\Sigma_2}(\mathcal{L}(G_{2}))\nonumber\\
\mathcal{L}_m(G_{12}) &= P^{-1}_{\Sigma_1 \cup \Sigma_2\to\Sigma_1}(\mathcal{L}_m(G_{1})) \cap P^{-1}_{\Sigma_1 \cup \Sigma_2\to\Sigma_2}(\mathcal{L}_m(G_{2})).\nonumber
\end{align}\qed
\end{definition}

\subsection{Supervisory Control Theory}
The Supervisory Control Theory is a formal method, based on language and automata theory, to the systematic calculus of supervisors \cite{W2014}. The system to be controlled is called \textit{plant}, the controller agent is called \textit{supervisor} and the control problem is to find a supervisor which enforces the specifications in a minimally restrictive way.
The plant is modeled by an automaton $G = (\_, \Sigma, \_, \_, \_)$ 
and $\Sigma = \Sigma_c \cup \Sigma_u$ where $\Sigma_c$ is the set of controllable events, which can be disabled by an external agent, and $\Sigma_u$ is the set of uncontrollable events, which cannot be disabled by an external agent. The plant represents the logical model of the DES, the system behavior under no control action. The supervisor's $S$ role is to regulate the plant behavior to meet a desired behavior $K$ disabling controllable events.
Let $ E $ be an automaton that represents the specification imposed on $ G $.
We say that $ K = \LANG_{m}(G \parallel E) \subseteq \LANG_{m}(G) $ is \textit{controllable} w.r.t. $G$ if $ \overline{K} \Sigma_{uc} \cap \LANG(G) \subseteq \overline{K} $.
A nonblocking supervisor $ V $ for $ G $ such that $ \LANG_{m}(V/G)= K $ exists if and only if $ K $ is controllable w.r.t. $ G $.
If $ K $ does not satisfy the condition, then the \textit{supremal controllable and  nonblocking sublanguage} $ Sup\, \mathcal{C} (K,G) $ can be synthesized.
It represents the least restrictive nonblocking supervisor.
For $G$ and $K$, a monolithic supervisor automaton $S $ can be computed to represent $ Sup\, \mathcal{C} (K,G) $ such that $ \LANG_{m} (S)= Sup\, \mathcal{C} (K,G) \subseteq K $.

The generated and marked language of a plant $G$ under the action of a supervisor $S$ are, respectively, $\mathcal{L}(S/G)$ and $\mathcal{L}_m(S/G) \subseteq \mathcal{L}(S/G)$.

The space explosion of the monolithic supervisor synthesis can be avoided using decentralized techniques, as the Local Modular Supervisory Control \cite{Queiroz2002} where one supervisor is synthesized for each specification, and each one of the supervisors has only a partial view of the plant. The \textit{global plant} $G$ is composed of $n$ sub-plants $H_i, i\in\{1 \ldots n\}$, such that their event sets $\Sigma_{H_i}$ are disjoint and $G = ||_{i = 1}^{n}H_i$, also the \textit{global specification} $E$ is composed of $m$ sub-specifications $E_j, j\in\{1\ldots m\}$, such that their event sets are represented by $\Sigma_{E_j}$ and $E = ||_{j = 1}^{m}E_j$. A \textit{local plant} $G_j$ is such that $G_j = ||_{a \in A_j}H_a$ with $A_j = \{ i \in \{1 \ldots n\} | \Sigma_{H_i} \cap \Sigma_{E_j} \neq \emptyset \}$.

In the Local Modular Control, the local supervisor $S_j = Sup\, \mathcal{C} (K_j,G_j)$, where $K_j = \LANG_{m}(G_j \parallel E_j)$. Each supervisor is nonblocking by construction, but their combined behaviour has to be nonblocking in order to have the same behavior than the monolithic control solution. To check if supervisors are nonconflicting, the equality in  (\ref{eq:nonconflictOriginal}) must be verified.

\begin{equation}\label{eq:nonconflictOriginal}
    ||^m_{j=1} \overline{\LM(S_j)} = \overline{\LM(||^m_{j=1} S_j)}.
\end{equation}

\subsection{Synchronizing Automata}

The original definition of a synchronizing automaton \cite{Volkov2008} is presented and the idea is extended to be used in the context of Supervisory Control Theory. A synchronizing deterministic finite automaton is a DFA that has a word that, when executed from any state of the automaton, leads to a known state.

\begin{definition}\label{def:SyncDFA}
\cite{Volkov2008} A complete automaton $G = (Q, \Sigma, \delta, \_,\_)$ is synchronizing if and only if for any pair of states $q,q' \in Q$ there exists a word $w\in \Sigma^*$, called \textit{synchronizing word}, such that $q.w = q'.w$, $\forall q,q'\in Q$. \qed
\end{definition}

A \textit{complete automaton} in the definition refers to an automaton with a complete transition function, that is, transitions labeled with all the events in the event set are available in each state. Also, the initial state is irrelevant to the original property, so it is intentionally omitted in the following example.

\begin{example}\label{ex:ConvetionalSA}  Consider the synchronizing automaton $A=(Q,\Sigma,\_,\_,\_)$ of Fig.~\ref{fig:syncOrig}. The word $w = ab^3ab^3a$ leads the automaton to state 1, regardless the origin state. Using the notation established before, $Q.w=1$, $Q=\{0,1,2,3\}$. It is straightforward that any word $sw$, $s \in \Sigma^*$, also leads the automaton to state 1. 

\begin{figure}[h]
    \centering
    \begin{center}
\begin{tikzpicture}[scale=0.13]
\tikzstyle{every node}+=[inner sep=0pt]
\draw [black] (28.3,-15.4) circle (3);
\draw (28.3,-15.4) node {$0$};
\draw [black] (45.9,-15.4) circle (3);
\draw (45.9,-15.4) node {$1$};
\draw [black] (28.3,-32.3) circle (3);
\draw (28.3,-32.3) node {$3$};
\draw [black] (45.9,-32.3) circle (3);
\draw (45.9,-32.3) node {$2$};
\draw [black] (31.3,-15.4) -- (42.9,-15.4);
\fill [black] (42.9,-15.4) -- (42.1,-14.9) -- (42.1,-15.9);
\draw (37.1,-15.9) node [below] {$a,b$};
\draw [black] (45.9,-18.4) -- (45.9,-29.3);
\fill [black] (45.9,-29.3) -- (46.4,-28.5) -- (45.4,-28.5);
\draw (45.4,-23.85) node [left] {$b$};
\draw [black] (42.9,-32.3) -- (31.3,-32.3);
\fill [black] (31.3,-32.3) -- (32.1,-32.8) -- (32.1,-31.8);
\draw (37.1,-31.8) node [above] {$b$};
\draw [black] (28.3,-29.3) -- (28.3,-18.4);
\fill [black] (28.3,-18.4) -- (27.8,-19.2) -- (28.8,-19.2);
\draw (28.8,-23.85) node [right] {$b$};
\draw [black] (46.567,-12.487) arc (194.8285:-93.1715:2.25);
\draw (51.11,-9.67) node [above] {$a$};
\fill [black] (48.62,-14.16) -- (49.52,-14.44) -- (49.26,-13.47);
\draw [black] (48.454,-33.851) arc (86.47119:-201.52881:2.25);
\draw (50.47,-38.49) node [below] {$a$};
\fill [black] (46.22,-35.27) -- (45.67,-36.04) -- (46.67,-36.1);
\draw [black] (28.356,-35.288) arc (28.79888:-259.20112:2.25);
\draw (24.5,-38.91) node [below] {$a$};
\fill [black] (25.96,-34.16) -- (25.02,-34.11) -- (25.5,-34.99);
\end{tikzpicture}
\end{center}
    \caption{Example \ref{ex:ConvetionalSA}- Conventional synchronizing automaton \cite{Volkov2008}.}
    \label{fig:syncOrig}
\end{figure}
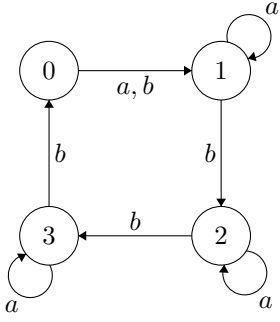

\qed
\end{example}
 
If the word $w$ is a synchronizing word, the operation $Q.w$ results in a singleton set. Also, the set of all synchronizing words of an automaton $G$ is denoted by $Syn(G)$: 
$$Syn(G)=\{w\in \Sigma^*| \; |Q.w|=1 \}.$$

\section{Problem Statement}\label{sec:prob}

Let $G$ be a manufacturing system modeled as a discrete event system under supervision of a supervisor $S$, obtained through Supervisory Control Theory, $S\subseteq {\cal L}_m(G)$.  Consider that such control system (plant and supervisor) lose synchronization due to a failure or an attack. In such a case, $\exists s\in \Sigma^*$, $\exists\sigma\in \Sigma$, such that $s\sigma\in{\cal L}_m(G)\cap S$, however, due to a failure or attack, the transition with $\sigma$ is not ``communicated'' to the supervisor  and the current state of the plant does not match the state estimated by the supervisor.  
Propose a method to resynchronize the control system (plant and supervisor), based on synchronizing automata. 

\section{Main Results}\label{sec:main}

The main idea of this work is to adapt the Supervisory Control Theory to deal with synchronizing automata, such that the features of this model can be used to solve desynchronization that may be caused by attacks or failures. In order to do so, we organize this section into four subsections. First we extend the concept of synchronizing automata to a more specific class that is the synchronizing automata w.r.t. the initial state and present some properties (Section \ref{sec:Basic}). Then, we present how synchronizing automata w.r.t. the initial state behave under some automata operations (Section \ref{sec:Operations}) and we show how these automata can be used in the context of Supervisory Control Theory (Section \ref{sec:SCT}). Finally, we present a method to model conventional DES problems as synchronizing automata w.r.t. the initial state (Section \ref{sec:Rec}).

\subsection{Basic Definitions}\label{sec:Basic}

When modeling a system, it is common to use partial transition functions and work with the language starting at the initial state. So, the idea of synchronization makes more sense when defined in relation to the initial state. In Definition \ref{def:SyncDFAIS} a new class of synchronizing automata is presented, the \textit{synchronizing automata w.r.t. the initial state}. In this new definition of synchronicity, the initial state cannot be omitted. %

\begin{definition}\label{def:SyncDFAIS}

An automaton $G = (Q, \Sigma, \_, q_0, \_)$ is synchronizing w.r.t. the initial state if there exists a word $w \in\Sigma^*$, called synchronizing word, such that $Q . w = \{q_0\}$. \qed

\end{definition}

In words, $G$ is a synchronizing automata w.r.t. the initial state if for any state $q\in Q$ of $G$, there is a word $w$ such that $q.w=q_0$. 
The set of synchronizing words w.r.t. the initial state of an automaton $G$ is represented by $Syn_{q_0}(G)$. In order to simplify the notation, we define that $Syn_{q_0}(G)=I$.

\begin{example} \label{ex:mi}
Let $A=(\_,\Sigma_1, \_, \_, \_)$ be an automaton with two states, in Fig.~\ref{fig:mi}. The word $w=c\in\Sigma_1^*$ is the shortest of the synchronizing words of $A$ and the automaton is a synchronizing automaton w.r.t. the initial state.

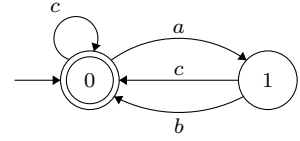
\begin{figure}[htbp]
     \begin{center}
        \begin{tikzpicture}[scale=0.13] \footnotesize
            \tikzstyle{every node}+=[inner sep=0pt]
            \draw [black] (25.7,-18.2) circle (3);
            \draw (25.7,-18.2) node {$0$};
            \draw [black] (25.7,-18.2) circle (2.4);
            \draw [black] (43.9,-18.2) circle (3);
            \draw (43.9,-18.2) node {$1$};
            \draw [black] (27.88,-16.151) arc (125.92933:54.07067:11.792);
            \fill [black] (41.72,-16.15) -- (41.37,-15.28) -- (40.78,-16.09);
            \draw (34.8,-13.41) node [above] {$a$};
            \draw [black] (41.412,-19.866) arc (-62.24811:-117.75189:14.199);
            \fill [black] (28.19,-19.87) -- (28.66,-20.68) -- (29.13,-19.8);
            \draw (34.8,-22) node [below] {$b$};
            \draw [black] (18,-18.2) -- (22.7,-18.2);
            \fill [black] (22.7,-18.2) -- (21.9,-17.7) -- (21.9,-18.7);
            \draw [black] (40.9,-18.2) -- (28.7,-18.2);
            \fill [black] (28.7,-18.2) -- (29.5,-18.7) -- (29.5,-17.7);
            \draw (34.8,-17.7) node [above] {$c$};
            \draw [black] (23.598,-16.076) arc (252.43495:-35.56505:2.25);
            \draw (22.13,-11.24) node [above] {$c$};
            \fill [black] (26.11,-15.24) -- (26.83,-14.63) -- (25.87,-14.33);
        \end{tikzpicture}
    \end{center}
    \caption{Example~\ref{ex:mi}: Synchronizing automaton $A$}
    \label{fig:mi}
\end{figure}
\qed
\end{example}

Proposition \ref{prop:lang} demonstrates some properties of synchronizing automata w.r.t. the initial state, regarding their languages, adapted from \cite{Volkov2008}.

\begin{proposition}\label{prop:lang}
Let $G = (\_, \Sigma, \_, q_0, \_)$ be a synchronizing automaton w.r.t. the initial state and $I\neq\emptyset$. Then:  
\begin{enumerate}
    \item[a)] $\LG(G)I\LG(G) \subseteq \LG(G)$;
    \item[b)] $\LG(G) I \LM(G) \subseteq \LM(G)$;
    \item[c)] $\LM(G) I \LM(G) \subseteq \LM(G)$.
\end{enumerate}
The language $\LM(G)$ is called a \textit{synchronizing language}.
\end{proposition}
\begin{proof}

Any word $s \in \LG(G)$ leads to a state $q$ ($q_0.s = q$) that, when followed by a word $w\in I$, reaches state $q_0$ ($q.w=q_0$), from Definition \ref{def:SyncDFA}. So, $\forall s\in  \LG(G)$, $\forall w\in I$ $q_0.sw = q_0$.

\begin{enumerate}
    \item[a)]Let $s\in\LG(G)$, and $w \in I$, then:
    \begin{equation}
    sw\in \LG(G)I\subseteq\LG(G)\nonumber
    \end{equation}
    and $q_0.sw=q_0$.
We also know that, for any automaton $G$, ${\cal L}_G(q_0)= \LG(G)$ and we can concatenate $\LG(G)$ to both sides and find: 
\begin{equation}
    sw\LG(G)\subseteq\LG(G)I\LG(G)\subseteq\LG(G)\LG(G).\label{eq:uw}
\end{equation}
(\ref{eq:uw}) becomes: 
\begin{equation}
   \LG(G)I\LG(G)\subseteq\LG(G).
\end{equation}
proving a).
 \item[b)] Let $u\in\LG(G)$, and $w \in I$, then:

\begin{equation}
    uw\in \LG(G)I\nonumber
\end{equation}
and $q_0.uw=q_0$.
We also know that, for any automaton $G$, ${\cal L}_G(q_0) \cap \LM(G)= \LM(G)$. So we can concatenate $\LM(G)$ to $uw$ and:

\begin{equation}
   uw\LM(G)\subseteq\LG(G)I\LM(G)\nonumber
\end{equation}
Since $\forall s\in \LG(G)I$, $q_0.s=q_0$, then:
\begin{equation}
   uw\LM(G)\subseteq\LG(G)I\LM(G)\subseteq\LM(G).\nonumber
\end{equation}
proving b).

 \item[c)]Given that:
\begin{equation}
    \LM(G) \subseteq \LG(G)\nonumber
\end{equation}
we have:
\begin{equation}
    \LM(G)I\LM(G)\subseteq\LM(G).\nonumber
\end{equation}
\end{enumerate}
\end{proof} 

A synchronizing automaton w.r.t. the initial state is synchronizing to any state if it is also accessible, given that it is always possible to lead any state to the initial state and then to any other state. 

\begin{corollary}\label{cor:block}
If $G = (Q, \_, \_, q_0, \_)$ is a synchronizing automaton w.r.t. the initial state and every state of $G$ is accessible then:
\begin{enumerate}
    \item[a)] $G$ is a synchronizing automaton w.r.t. any state $q'\in Q$;
    \item[b)] $G$ is coaccessible.
\end{enumerate} 
\end{corollary}
\begin{proof}
If $G$ is synchronizing w.r.t. the initial state then there is a set $I\neq\emptyset$ such that $$\LG(G) I \LG(G) \subseteq \LG(G),$$ from (Proposition \ref{prop:lang}). 
If $G$ is accessible, for every state $q\in Q$ there is at least a word $u \in \LG(G)$ such that:
\begin{equation}\label{eq:delta}
 q_0.u=q 
\end{equation} 
and since $I\neq\emptyset$, from any state $q'\in Q$, $q'.w=q_0$, with $w \in I$.  From (\ref{eq:delta}), we know that $q'.wu=q$. Then, $Q.wu=\{q\}$, $wu\in I$ and $G$ is synchronizing w.r.t. state $q$, showing item a). If $G$ is accessible, every state $q\in Q_m\subseteq Q$ is reachable from the initial state, $q_0.u=q$, with $u\in\LM(G)$. If $G$ is synchronizing, then there exists $w\in I$ such that $q.w=q_0$ and from $q_0$ all states are reachable. Then, we can conclude that $G$ is coaccessible, showing item b).
\end{proof}

\subsection{Operations with Synchronizing Automata}\label{sec:Operations}

In general, it makes little sense, in the Supervisory Control Theory, to expect a supervisor to be synchronizing when the automata that originate that supervisor are not. So, our strategy is to model the system and specification as synchronizing automata and see under what conditions the synchronization word survives the synthesis procedure. In this context, it is important to analyze how the synchronizing word survives the parallel composition of synchronizing automata.

\begin{lemma}\label{lem:proj}
Let $L \subseteq \Sigma_1^*$ be a synchronizing language and $I$ be the set of all synchronizing words w.r.t. the initial state of $L$ and $P:\Sigma^* \to \Sigma_1^*$, $\Sigma_1\subseteq\Sigma$ then $P^{-1}(L)$ is a synchronizing language.

\begin{proof}To show that $P^{-1}(L)$ is a synchronizing language, we must show that $\exists I_K$, such that $P^{-1}(L) I_K P^{-1}(L) \subseteq P^{-1}(L)$. Since $L$ is a synchronizing language, then $L I L \subseteq L$ (Proposition \ref{prop:lang}). Applying the inverse projection to both sides: 
$$P^{-1}(L I L) \subseteq P^{-1}(L) $$
We can decompose the left side of the expression, resulting in:
$$ P^{-1}(L) P^{-1}(I) P^{-1}(L) \subseteq P^{-1}(L)$$
replacing $P^{-1}(I)=I_K$ and $P^{-1}(L)=B$ we have:
$$ BI_KB \subseteq B. $$
So, $B= P^{-1}(L)$ is a synchronizing language.

\end{proof}

\end{lemma}

In the context of automata, the inverse projection creates self-loops in all states for each symbol in $\Sigma_2 \setminus \Sigma_1$. It is easy to see that this operation does not turn a synchronizing automaton unsynchronizing, but only increases the number of synchronizing words.

\begin{example} \label{ex:proj}
Let $A=(\_,\Sigma_1, \_,\_,\_)$ be the synchronizing automaton w.r.t the initial state previously presented in Fig.\ref{fig:mi}, $\Sigma_1=\{a,b,c\}$ and $\Sigma=\{a,b,c,x\}$. Consider the natural projection $P:\Sigma\to \Sigma_1$. In Fig.\ref{fig:proj}, an automaton that models the language $P^{-1}(\LM(A))$ is shown. It is easy to see that any word in $I_P=(\Sigma \setminus \Sigma_1)^* c (\Sigma \setminus \Sigma_1)^*c^*(\Sigma \setminus \Sigma_1)^*=x^*cx^*c^*x^*\in\Sigma^*$ is a synchronizing word w.r.t. the initial state of the resulting automaton.
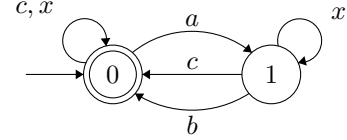
\begin{figure}[h]
    \centering
\begin{center}
\begin{tikzpicture}[scale=0.13]
\tikzstyle{every node}+=[inner sep=0pt]
\draw [black] (29.7,-15.4) circle (3);
\draw (29.7,-15.4) node {$0$};
\draw [black] (29.7,-15.4) circle (2.4);
\draw [black] (45.9,-15.4) circle (3);
\draw (45.9,-15.4) node {$1$};
\draw [black] (26.969,-14.188) arc (273.80557:-14.19443:2.25);
\draw (21.7,-9.72) node [above] {$c,x$};
\fill [black] (29,-12.49) -- (29.45,-11.66) -- (28.45,-11.73);
\draw [black] (42.9,-15.4) -- (32.7,-15.4);
\fill [black] (32.7,-15.4) -- (33.5,-15.9) -- (33.5,-14.9);
\draw (37.8,-14.9) node [above] {$c$};
\draw [black] (43.632,-17.348) arc (-57.30124:-122.69876:10.795);
\fill [black] (31.97,-17.35) -- (32.37,-18.2) -- (32.91,-17.36);
\draw (37.8,-19.56) node [below] {$b$};
\draw [black] (31.69,-13.171) arc (129.33557:50.66443:9.64);
\fill [black] (43.91,-13.17) -- (43.61,-12.28) -- (42.97,-13.05);
\draw (37.8,-10.49) node [above] {$a$};
\draw [black] (20.8,-15.4) -- (26.7,-15.4);
\fill [black] (26.7,-15.4) -- (25.9,-14.9) -- (25.9,-15.9);
\draw [black] (46.614,-12.498) arc (193.89909:-94.10091:2.25);https://pt.overleaf.com/project/5b17bef40a021d6739b48136
\draw (52.8,-9.74) node [above] {$x$};
\fill [black] (48.64,-14.2) -- (49.53,-14.5) -- (49.29,-13.52);
\end{tikzpicture}
\end{center}
    \caption{Example~\ref{ex:proj}: Synchronizing automaton w.r.t. the initial state, modeling the inverse projection of the machine $A$ to $\Sigma$.}
    \label{fig:proj}
\end{figure}
\qed
\end{example}

From the definition of the inverse projection, we can conclude that $I \subset I_K$, so every synchronizing word of $L$ is also a synchronizing word of $P^{-1}(L)$.

Now we deal with the intersection operator.

\begin{lemma}\label{lem:inter}
Let $L_1, L_2\subseteq\Sigma^*$ be synchronizing languages. Let $I_1, I_2$ be, respectively, the sets of synchronizing words of their original automata. If $I_1\cap I_2\neq \emptyset$, then the nonempty language $K = L_1 \cap L_2$ is a synchronizing language and its automaton is synchronizing w.r.t. the initial state.
\end{lemma}
\begin{proof}
From Proposition \ref{prop:lang}, we have that:
\begin{equation}\label{eq:l1il1}
    L_1 I_1 L_1 \subseteq L_1
\end{equation}
\begin{equation}\label{eq:l2il2}
    L_2 I_2 L_2 \subseteq L_2.
\end{equation}
for all $s \in L_1 \cap L_2$ and $w \in I_1 \cap I_2$ is straightforward that $s w s \in L_1 I_1 L_1$ and $s w s \in L_2 I_2 L_2$, so:
\begin{equation}\label{eq:sws}
    sws\in  L_1I_1L_1\cap L_2I_2L_2.
\end{equation}
(\ref{eq:sws}) can be rewritten as:
\begin{equation}
    sws\in  (L_1\cap L_2)(I_1\cap I_2)(L_1\cap L_2)\subseteq L_1\cap L_2.
\end{equation}
then, $(L_1\cap L_2)I_1\cap I_2(L_1\cap L_2)\subseteq L_1\cap L_2$ and $L_1\cap L_2$ is a synchronizing language.

\end{proof}

Using the last two lemmas, it is possible to define conditions under which the parallel composition maintains the synchronicity of the original synchronizing automata. This result is presented in Proposition~\ref{thm:comp}.

\begin{proposition}

\label{thm:comp}
Let $G_1 = (Q_1, \Sigma_1, \delta_1, q_{01}, Q_{m1})$ and $G_2 = (Q_2, \Sigma_2, \delta_2, q_{02}, Q_{m2})$ be synchronizing automata w.r.t. the initial state and let $\Sigma = \Sigma_1 \cup \Sigma_2$. The resulting automaton $G = G_1 || G_2$ is synchronizing w.r.t. the initial state if $P^{-1}_{\Sigma \to \Sigma_1}(I_1) \cap P^{-1}_{\Sigma \to \Sigma_2}(I_2) \neq \emptyset$, with $P_{\Sigma \to \Sigma_i}:\Sigma^*\rightarrow \Sigma_i^*$, $i\in\{1,2\}$.
\end{proposition}

\begin{proof} Because $G_1$ and $G_2$ are synchronizing  w.r.t. the initial state, we know, using Lemma~\ref{lem:proj}, that the languages $P^{-1}_{\Sigma \to \Sigma_1}(\LM (G_1))$ and $P^{-1}_{\Sigma \to \Sigma_2}(\LM (G_2))$ are also synchronizing w.r.t. the initial state.

Considering that $P^{-1}_{\Sigma \to \Sigma_1}(I_1) \cap P^{-1}_{\Sigma \to \Sigma_2}(I_2) \neq \emptyset$ and also that:
$$\LM(G_1||G_2) = P^{-1}_{\Sigma \to \Sigma_1}(\LM(G_{1}))\cap P^{-1}_{\Sigma \to \Sigma_2}(\LM(G_{2}))$$
we have, using Lemma~\ref{lem:inter}, that the language $\LM(G_1||G_2)$ is a synchronizing language and that $G_1||G_2$ is a synchronizing automaton w.r.t. the initial state.
\end{proof}
The next step is to analyze how synchronizing words behave in the synthesis of controllable and nonblocking supervisors using the Supervisory Control Theory.

\subsection{Supervisory Control Theory with Synchronizing Automata}\label{sec:SCT}

The synthesis of a supervisor has 3 main steps: model the open loop behavior and specifications; compute the desired language; synthesize the supremal controllable and nonblocking sublanguage. At this point, we assume that we are able to model and specify using synchronizing automata (we present how in Section \ref{sec:model}).

In the following, we show under what conditions we have a synchronizing language as a result of the synthesis.

\begin{theorem}\label{thm:sct}
Let $G$ be a plant and let $E$ be a specification, both modeled as synchronizing automata w.r.t. the initial state. A nonempty controllable and nonblocking supervisor $S$ such that $S=Sup{\cal C}(K,G)$, $K=G||E$, 
is a synchronizing automata w.r.t. the initial state if $\Sigma_{u}^* \cap I \neq \emptyset$, with $I = Syn_{q_0}(G \parallel E)$.

\end{theorem}
\begin{proof}
Let $K = G \parallel E = (Q, \Sigma, \_, q_0, Q_m)$ and $S = (Q_s, \Sigma, \_, q_0, Q_{ms})$, where $Q_s \subseteq Q$ and $Q_{ms}=Q_m\cap Q_s$. In the sense of controllability, every state $q_f \in Q \setminus Q_s$ is a bad state, because fails the principle of controllability.

Since $\Sigma_{u}^* \cap I \neq \emptyset$, there is at least a word $w = \sigma_1 \sigma_2 \ldots \sigma_n \in \Sigma_{u}^* \cap I$. There are two possibilities to be considered.
\begin{enumerate}
    \item[a)]  the trace $w$ executed from any state $q\in Q$ does not visit a bad state:\\
    If this is the case, since all states of $K$ that are visited are good states, they will be kept in $S$. So, $w \in I$ is, also, a synchronizing word of $S$. 
    \item[b)]  the trace $w$ executed from any state $q\in Q$ visits a bad state;\\
    When obtaining $S$, states of the automaton that implements $K$ are removed, if they are bad states. If there is a word $w = \sigma_1 \sigma_2 \ldots \sigma_n \in \Sigma_{u}^* \cap I$, where $I = Syn_{q_0}(K)$, then every $q' \in Q$ where $q' \xrightarrow{\sigma_1 ... \sigma_p} q_f \xrightarrow{\sigma_{p+1} ... \sigma_n} q_0$ is also a bad state and is not in $Q_s$, so every state $q$ that leads to a bad state, using uncontrollable events, is also removed, then $w$ is completely removed, indicating that $\LG(S) = \{\epsilon\}$ and the supervisor is empty.
\end{enumerate}

After the bad states removal, the accessible part of the resulting automaton is always coaccessible, using Corollary~\ref{cor:block}, then nonblocking. 
\end{proof}

It is straightforward to apply Theorem \ref{thm:sct} to the Local Modular Supervisory Control of DES \cite{Queiroz2002}.

\begin{corollary} \label{cor:loc}
Let $E_j$ be the local specifications of the system and $G_j=(\_,\Sigma_j,\_,\_,\_)$, be the local plants, with $\Sigma_{uj}\subseteq\Sigma_{j}$ as the uncontrollable events of $G_j$  and $j \in \{1 \ldots m\}$. If $G_j$ and $E_j$ are synchronizing automata w.r.t. the initial state, $\Sigma_{uj}^* \cap I_j \neq \emptyset$, $I_j = Syn_{q_0}(G_j \parallel E_j)$ then the local supervisors $S_j$ are also synchronizing automata w.r.t. the initial state.\end{corollary}
\begin{proof}
This results follows from the direct application of Theorem \ref{thm:sct} to local specifications and local plants. 
\end{proof}

Corollary~\ref{cor:loc_block} shows that synchronizing local supervisors are nonconflicting.

\begin{corollary} \label{cor:loc_block}
Let $S_j$ be the local supervisors of a system, defined as synchronizing automata w.r.t. the initial state, then these modular supervisors are nonconflicting, so $\overline{\LM(||^m_{j=1} S_j)} = ||_{j=1}^{m} \overline{\LM(S_j)}$.
\end{corollary}

\begin{proof}
If $S_j$ is synchronizing w.r.t. the initial state, from  Corollary~\ref{cor:block}, we know $S_j$ is coaccessible and then:

\begin{equation}\label{eq:closure}
    \overline{\LM(S_j)} = \LG(S_j).
\end{equation}

From Theorem~\ref{thm:comp}, we know that $S = ||^m_{j=1} S_j$ is a synchronizing automaton w.r.t. the initial state and is, also, coaccessible, such that:

\begin{equation}
    \overline{\LM(S)} = \LG(S)\nonumber
\end{equation}

replacing $S$ by $||^m_{j=1} S_j$ on both sides:

\begin{equation}
    \overline{\LM(||^m_{j=1} S_j)} = ||^m_{j=1}\LG(S_j).\nonumber
\end{equation}
Using (\ref{eq:closure}) we have

\begin{equation}
    \overline{\LM(||^m_{j=1} S_j)} = ||_{j=1}^{m} \overline{\LM(S_j)}.\nonumber
\end{equation}
So, the supervisors $S_j$ are nonconflicting.
\end{proof}

In the next section synchronization concepts presented so far are used to implement a recovery procedure for a classical SCT problem.

\subsection{Synchronization using Recovery Events}\label{sec:Rec}

In order to integrate the idea of synchronization with the Supervisory Control Theory, we propose the creation of a recovery event that connects each state of the plant to the initial state, including a self-loop in the initial state. Also, if the specification is of a buffer type, we create a recovery event to have the buffer move from any state to the initial state. The same idea can be applied to any other type of specification. 

It is important to note that the creation of the recovery events in a system is only possible when the components of the system admit a restart procedure regardless their current state. This restart can be automatic, when the system has a built-in reset, or manual, when an operator has to manually restart the system. Although the existence of a restart procedure is common in many industrial devices, some systems may not 
be restarted due to physical restrictions, for instance, systems with slow dynamics in which abrupt changes are not possible (power systems, thermal systems, and so on).

A procedure that turns plant and specification into synchronizing automata is presented next.

\subsubsection{Modeling}\label{sec:model}

Consider a system composed of machines $M_i' = (Q_i, \Sigma_i', \delta_i', q_{0i}, \_)$, $i\in\{1\ldots m\}$, and buffer specifications $B_j' = (Q_j, \Sigma_j', \delta_j', q_{0j}, \_)$, $j \in\{ 1 \ldots n\}$. To turn the automata into synchronizing automata, the procedure is:

\begin{enumerate}
    \item[a)] For each plant $M_i'$ we redefine it to $M_{i} = (Q_i, \Sigma_{i}, \delta_{i}, q_{0i}, \_)$ where $\Sigma_{i} = \Sigma_i' \cup \Sigma_{ri}$, $\Sigma_{ri}=\{ r_i \}$, and $\delta_{i}$ as:
    
    \begin{equation*}
        \delta_{i}(q, \sigma) =  
        \begin{cases} 
           q_{0i} & \text{if } \sigma = r_i \\
           \delta_i'(q, \sigma)  & \text{if } \sigma \neq r_i.
        \end{cases}
    \end{equation*}
  
  \item[b)] For each buffer specification $B_j'$ we redefine it to $B_{j} = (Q_j, \Sigma_{j}, \delta_{j}, q_{0j}, \_)$ where $\Sigma_{j} = \Sigma_j' \cup \{ r_{B_j} \}$ and $\delta_{j}$ as:
  
    \begin{equation*}
        \delta_{j}(q, \sigma) =  
        \begin{cases} 
           q_{0j} & \text{if } \sigma = r_{B_j} \\
           \delta_j'(q, \sigma)  & \text{if } \sigma \neq r_{B_j}.
        \end{cases}
    \end{equation*}
\end{enumerate}
\subsubsection{Synthesis}

We propose two modifications to the Supervisory Control Theory, related to the verification of controllability and nonblockingness under a new partition of the event set. Instead of partitioning the event set into controllable and uncontrollable events, we use a third set of events that carries the recovery events, as in \cite{Shu2014}. This modification is justified because in the controllability analysis we need the recovery events to behave as uncontrollable events, but we do not desire that the recovery events take part on the blocking analysis, because the system may be blocking and this will be detected only if recovery events are disregarded.

\begin{definition}

Let $G = (\_, \Sigma, \_,\_, \_)$ be a deterministic finite automaton, synchronizing w.r.t. the initial state, and let $\Sigma = \Sigma_c \cup \Sigma_u \cup \Sigma_r$, with $\Sigma_c$ as the controllable events set, $\Sigma_u$ as the uncontrollable events set and $\Sigma_r$ as the recovery events set. Any word of cardinality $n$ formed as an arrangement, without repetition, of the set $\Sigma_r$, with $n = |\Sigma_r|$ is a synchronizing word of $G$. \qed

\end{definition}

As established in Proposition~\ref{thm:comp}, the parallel composition of two synchronizing automata, w.r.t. the initial state, is also synchronizing when the intersection between their sets of synchronizing words, when inverse projected to the complete event set, is nonempty. Such intersection always exists when using recovery events, as defined in Section~\ref{sec:model}.

\begin{corollary}
Let $G_1 = (\_, \Sigma_1, \_, \_,\_)$ and $G_2 = (\_, \Sigma_2,\_,$ $\_,\_)$ be synchronizing w.r.t. the initial state and $\Sigma = \Sigma_1 \cup \Sigma_2$ and $\Sigma_r = \Sigma_{r1} \cup \Sigma_{r2}$, where $\Sigma_r \subset \Sigma$, $\Sigma_{r1} \subset \Sigma_1$ and $\Sigma_{r2} \subset \Sigma_2$. Also, let $I_1$ and $I_2$ be the sets of synchronizing words of $G_1$ and $G_2$, respectively. The resulting automaton $G = G_1 || G_2$ is synchronizing w.r.t. the initial state.
\end{corollary}

\begin{proof}
To show that automaton $G = G_1 || G_2$ is synchronizing w.r.t. the initial state it is enough to show that $P^{-1}_{\Sigma \to \Sigma_1}(I_1) \cap P^{-1}_{\Sigma \to \Sigma_2}(I_2)\neq\emptyset$, (Proposition \ref{thm:comp}), $P_{\Sigma \to \Sigma_i}:\Sigma^*\to\Sigma_i^*$, $i\in\{1,2\}$.

Let $perm(\Sigma_a,\Sigma_b) = \{ s : s \in \Sigma_a^* \land \forall \sigma \in \Sigma_b, |P_{\Sigma_a \to \{\sigma\}}(s)| = 1\}$ be a subset of a $\Sigma_a^*$ where every event in $\Sigma_b$ occurs only once. When $\Sigma_a = \Sigma_b$ the resulting language carries the words that are permutations of the events in $\Sigma_a$.

By definition, 
$$perm(\Sigma_{r1}, \Sigma_{r1}) \subseteq I_1$$
$$perm(\Sigma_{r2}, \Sigma_{r2}) \subseteq I_2$$
and 
\begin{multline}perm(\Sigma_r, \Sigma_r) \subseteq perm(\Sigma_r, \Sigma_{r1}) \\ \subseteq P^{-1}_{\Sigma \to \Sigma_1}(perm(\Sigma_{r1},\Sigma_{r1}))
\label{eq:perm1}
\end{multline}
\begin{multline}
perm(\Sigma_r, \Sigma_r) \subseteq perm(\Sigma_r, \Sigma_{r2}) \\ \subseteq P^{-1}_{\Sigma \to \Sigma_2}(perm(\Sigma_{r2},\Sigma_{r2})).
\label{eq:perm2}\end{multline}

If $\Sigma_r\neq\emptyset$, we have that $perm(\Sigma_r, \Sigma_r)\neq\emptyset$. From (\ref{eq:perm1}) and (\ref{eq:perm2}), $perm(\Sigma_r, \Sigma_r) \subseteq P^{-1}_{\Sigma \to \Sigma_1}(I_1) \cap P^{-1}_{\Sigma \to \Sigma_2}(I_2)$. Then,  we can say that $G=G_1||G_2$, modeled with recovery events, is synchronizing w.r.t. the initial state.

\end{proof}

Now, it is necessary to redefine nonblockingness and controllability, since there is a new partition to the events set (including $\Sigma_r$). When verifying nonblocking, recovery events are ignored  because a blocking behavior should not be turned into nonblocking by recovery events. When verifying controllability, the recovery events should be considered as uncontrollable events, because a recovery event should never be disabled by the supervisor.

The definition of nonblocking for systems with recovery events is given in Definition \ref{def:nonbl}.
\begin{definition}\label{def:nonbl}
Let $G = (\_, \Sigma, \_, \_, \_)$ be a deterministic finite automata, synchronizing w.r.t. the initial state, with $\Sigma_r\subseteq\Sigma$. $G$ is nonblocking if:
$$\LG(G)\cap (\Sigma\setminus \Sigma_r)^* = \overline{\LM(G) \cap (\Sigma\setminus \Sigma_r)^*} $$ 
\qed
\end{definition}

The modified definition of controllability is presented in Definition \ref{def:control}.

\begin{definition}\label{def:control}

Let $G = (\_, \Sigma, \_, \_, \_)$ be a deterministic finite automata, synchronizing w.r.t. the initial state, so its event set can be partitioned into $\Sigma = \Sigma_c \cup \Sigma_u \cup \Sigma_r$. A language $K \subseteq \LM(G)$ is controllable if:
$$ \overline{K} (\Sigma_u \cup \Sigma_r) \cap \LG(G) \subseteq \overline{K}. $$
\qed
\end{definition}

If $K$ satisfies the condition, then it is controllable, otherwise the \textit{supremal controllable and  nonblocking sublanguage} can be synthesized. Theorem \ref{thm:sct} is reformulated in Corollary \ref{cor:sct}.

\begin{corollary}\label{cor:sct}
Let $G=||_{i=1}^m M_i$ and $E=||_{j=1}^n B_j$, modeled as synchronizing automata w.r.t. the initial state, as in Section \ref{sec:model}. Let $I$ be the set of synchronizing words of $K=G||E$, if  $(\Sigma_{u}\cup\Sigma_r)^* \cap I \neq\emptyset$ then a nonempty controllable and nonblocking supervisor $S$, synthesized from $G$ and $E$ is, also, a synchronizing automata w.r.t. the initial state.
\end{corollary}

\begin{proof}
Using the modeling approach proposed in Section \ref{sec:model}, recovery events are included in the subsystems and specifications such that: 
\begin{equation}
    I\cap\Sigma_r^*\neq\emptyset.\label{eq:sigmar}
\end{equation} 

From Definition~\ref{def:control} we know that the recovery events cannot be disabled by the supervisor (if that happened the controllability test would fail). Then, Theorem \ref{thm:sct} is valid replacing $\Sigma_u$ with $(\Sigma_u\cup\Sigma_r)$ in the statement and in the proof. 

In such a case, the condition for the validity of the Theorem is changed to $(\Sigma_{u}\cup\Sigma_r)^* \cap I \neq\emptyset$. From (\ref{eq:sigmar}) we know that $(\Sigma_{u}\cup\Sigma_r)^* \cap I \neq\emptyset$, so the condition is fulfilled and the supervisor is synchronizing w.r.t. the initial state.
\end{proof}

With this approach, a controllable and nonblocking supervisor is always synchronizing w.r.t. the initial state. 

The same approach can be applied to Local Modular Supervisory Control. As presented in Corollary~\ref{cor:loc}, each  local supervisor is controllable and nonblocking. However, even if the original system is nonconflicting, the nonconflict test over the supervisors with recovery events is necessary. Since recovery events do not take part into the nonblockingness verification, the nonconflicting test of (\ref{eq:nonconflictOriginal}) has to be adapted to ignore recovery events (Definition \ref{def:nonconflictReset}).

\begin{definition} \label{def:nonconflictReset}
Let $S_j$ be the local supervisors of a system, defined as synchronizing automata w.r.t. the initial state with event set $\Sigma_j = \Sigma_{cj} \cup \Sigma_{uj} \cup\Sigma_{rj}$. These supervisors are nonconflicting if:

\begin{equation}\label{eq:nonconflictReset}
||^m_{j=1} \overline{\LM(S_j) \cap (\Sigma_j \setminus \Sigma_{rj})^*}  =
\overline{\LM(||^m_{j=1} S_j) \cap (\Sigma\setminus \Sigma_r)^*}    \nonumber
\end{equation}
\qed
\end{definition}

In manufacturing systems, the recovery events typically share a transition with uncontrollable events in the plants, or are in self-loops. The resulting modular supervisors, when we remove the recovery events, are equal to the modular supervisors of the system when modeled without recovery events. When this is the case, a nonconflicting control system will be nonconflicting after the recovery events are added.

Next section shows a complete example of the application of synchronizing automata w.r.t. the initial state using recovery events to recover from a fault when the plant and the supervisor become unsynchronized.

\section{Case Studies}\label{sec:example}
In this section, we show how to model regular DES problems (the extended small factory \cite{W2014} and the Flexible Manufacturing System \cite{Queiroz2005}) as synchronizing automata and apply the reset procedure proposed in this paper.

\subsection{Extended Small Factory} 


Consider an extended version of the small factory, composed of three machines and two unity buffers, Fig.~\ref{fig:sf}. 

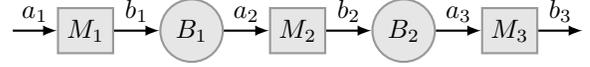
\begin{figure}[htbp]
\centering

\tikzstyle{buffer}=[circle,
                                    thick,
                                    minimum size=0.2cm,
                                    draw=gray!80,
                                    fill=gray!20]
                                    
\tikzstyle{machine}=[rectangle,
                                    thick,
                                    minimum size=0.6cm,
                                    draw=gray!80,
                                    fill=gray!20]
                                    
\tikzstyle{IO}=[rectangle,
                                    thick,
                                    minimum size=0.1cm,
                                   draw=none,fill=none]

\begin{tikzpicture}[>=latex,text height=1.5ex,text depth=0.25ex]

  \matrix[row sep=0.5cm,column sep=0.3cm] {
    & \node (I) [IO,draw=none]{$ $}; &
    & \node (m1) [machine]{$M_1$}; &
    & \node (b1)   [buffer]{$B_1$};   &
    & \node (m2) [machine]{$M_2$}; &
    & \node (b2)   [buffer]{$B_2$};   &
    & \node (m3) [machine]{$M_3$}; &
    & \node[draw=none] (O) [IO]{$ $}; &
        \\
    };
    
    \path[->]
        (I)  edge[thick,left] node[above] {$a_1$} (m1)
        (m1) edge[thick,left] node[above] {$b_1$} (b1)
        (b1) edge[thick,left] node[above] {$a_2$} (m2)
        (m2) edge[thick,left] node[above] {$b_2$} (b2)
        (b2) edge[thick,left] node[above] {$a_3$} (m3)
        (m3) edge[thick,left] node[above] {$b_3$} (O)
        ;
\end{tikzpicture}
            
\caption{Extended small factory diagram.}
\label{fig:sf}          
\end{figure}
\noindent{Originally, each machine is modeled by an automaton $M_i'=(\_,\Sigma_i,\_,\_,\_)$, $i\in\{1,2,3\}$, with 2 states (idle and working) and 2 transitions (start and finish). The unity buffers are also modeled with automata $B_j'$, $j\in\{1,2\}$, with two states and two transitions (Figure \ref{fig:autoOrig}).}

\begin{figure}[htbp]
    \centering
\subfloat[$M_i'$, $i\in\{1,2,3\}$]{          
\begin{tikzpicture}[scale=0.10] \footnotesize
            \tikzstyle{every node}+=[inner sep=0pt]
            \draw [black] (25.7,-18.2) circle (3);
            \draw (25.7,-18.2) node {$0$};
            \draw [black] (25.7,-18.2) circle (2.4);
            \draw [black] (43.9,-18.2) circle (3);
            \draw (43.9,-18.2) node {$1$};
            \draw [black] (27.88,-16.151) arc (125.92933:54.07067:11.792);
            \fill [black] (41.72,-16.15) -- (41.37,-15.28) -- (40.78,-16.09);
            \draw (34.8,-13.41) node [above] {$a_i$};
            \draw [black] (41.412,-19.866) arc (-62.24811:-117.75189:14.199);
            \fill [black] (28.19,-19.87) -- (28.66,-20.68) -- (29.13,-19.8);
            \draw (34.8,-22) node [below] {$b_i$};
            \draw [black] (18,-18.2) -- (22.7,-18.2);
            \fill [black] (22.7,-18.2) -- (21.9,-17.7) -- (21.9,-18.7);
        \end{tikzpicture}\hspace{1cm}
}
\subfloat[$B_j'$, $j\in\{1,2\}$]{          
\begin{tikzpicture}[scale=0.10] \footnotesize
            \tikzstyle{every node}+=[inner sep=0pt]
            \draw [black] (25.7,-18.2) circle (3);
            \draw (25.7,-18.2) node {$E$};
            \draw [black] (25.7,-18.2) circle (2.4);
            \draw [black] (43.9,-18.2) circle (3);
            \draw (43.9,-18.2) node {$F$};
            \draw [black] (27.88,-16.151) arc (125.92933:54.07067:11.792);
            \fill [black] (41.72,-16.15) -- (41.37,-15.28) -- (40.78,-16.09);
            \draw (34.8,-13.41) node [above] {$b_i$};
            \draw [black] (41.412,-19.866) arc (-62.24811:-117.75189:14.199);
            \fill [black] (28.19,-19.87) -- (28.66,-20.68) -- (29.13,-19.8);
            \draw (34.8,-22) node [below] {$a_{i+1}$};
            \draw [black] (18,-18.2) -- (22.7,-18.2);
            \fill [black] (22.7,-18.2) -- (21.9,-17.7) -- (21.9,-18.7);
        \end{tikzpicture}
}\caption{Original automata modeling the extended small factory.}\label{fig:autoOrig}
\end{figure}
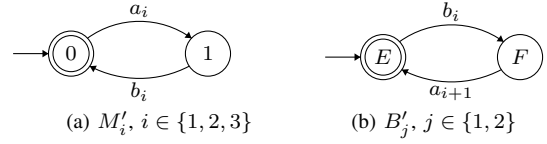

\noindent{First, the automata of Fig.~\ref{fig:autoOrig} are transformed into synchronizing automata w.r.t. the initial state,} by adding recovery events $r_i$ and $r_{B_j}$ to the models, such that each automaton is brought to the initial state when their synchronizing word is executed, as described in Section \ref{sec:model}.\\

\noindent{In Fig.~\ref{fig:model} the model of each part of the system and the shortest synchronizing word of each machine are shown. The shortest synchronizing word is the trace $r_i\in\Sigma^*_{ri}$ for each plant and $r_{B_j}\in\Sigma^*_{B_j}$ for each specification.}

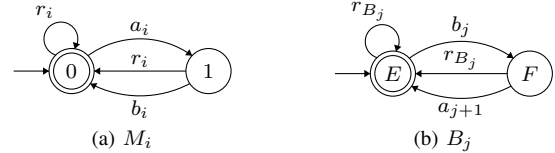
\begin{figure}[htbp]
    \centering
    \subfloat[$M_i$]{
        \begin{tikzpicture}[scale=0.10] \footnotesize
            \tikzstyle{every node}+=[inner sep=0pt]
            \draw [black] (25.7,-18.2) circle (3);
            \draw (25.7,-18.2) node {$0$};
            \draw [black] (25.7,-18.2) circle (2.4);
            \draw [black] (43.9,-18.2) circle (3);
            \draw (43.9,-18.2) node {$1$};
            \draw [black] (27.88,-16.151) arc (125.92933:54.07067:11.792);
            \fill [black] (41.72,-16.15) -- (41.37,-15.28) -- (40.78,-16.09);
            \draw (34.8,-13.41) node [above] {$a_i$};
            \draw [black] (41.412,-19.866) arc (-62.24811:-117.75189:14.199);
            \fill [black] (28.19,-19.87) -- (28.66,-20.68) -- (29.13,-19.8);
            \draw (34.8,-22) node [below] {$b_i$};
            \draw [black] (18,-18.2) -- (22.7,-18.2);
            \fill [black] (22.7,-18.2) -- (21.9,-17.7) -- (21.9,-18.7);
            \draw [black] (40.9,-18.2) -- (28.7,-18.2);
            \fill [black] (28.7,-18.2) -- (29.5,-18.7) -- (29.5,-17.7);
            \draw (34.8,-17.7) node [above] {$r_i$};
            \draw [black] (23.598,-16.076) arc (252.43495:-35.56505:2.25);
            \draw (22.13,-11.24) node [above] {$r_i$};
            \fill [black] (26.11,-15.24) -- (26.83,-14.63) -- (25.87,-14.33);
        \end{tikzpicture}
        }\hspace{1cm}
\subfloat[$B_j$]{            \begin{tikzpicture}[scale=0.1] \footnotesize
            \tikzstyle{every node}+=[inner sep=0pt]
            \draw [black] (25.7,-18.2) circle (3);
            \draw (25.7,-18.2) node {$E$};
            \draw [black] (25.7,-18.2) circle (2.4);
            \draw [black] (43.9,-18.2) circle (3);
            \draw (43.9,-18.2) node {$F$};
            \draw [black] (27.88,-16.151) arc (125.92933:54.07067:11.792);
            \fill [black] (41.72,-16.15) -- (41.37,-15.28) -- (40.78,-16.09);
            \draw (34.8,-13.41) node [above] {$b_j$};
            \draw [black] (41.412,-19.866) arc (-62.24811:-117.75189:14.199);
            \fill [black] (28.19,-19.87) -- (28.66,-20.68) -- (29.13,-19.8);
            \draw (34.8,-22) node [below] {$a_{j+1}$};
            \draw [black] (18,-18.2) -- (22.7,-18.2);
            \fill [black] (22.7,-18.2) -- (21.9,-17.7) -- (21.9,-18.7);
            \draw [black] (40.9,-18.2) -- (28.7,-18.2);
            \fill [black] (28.7,-18.2) -- (29.5,-18.7) -- (29.5,-17.7);
            \draw (34.8,-17.7) node [above] {$r_{B_j}$};
            \draw [black] (23.598,-16.076) arc (252.43495:-35.56505:2.25);
            \draw (22.56,-11.26) node [above] {$r_{B_j}$};
            \fill [black] (26.11,-15.24) -- (26.83,-14.63) -- (25.87,-14.33);
        \end{tikzpicture}
        }
    \caption{Synchronizing automaton $M_i$, $i \in \{1, 2, 3\}$, with $w_i = r_i\in I_1$, and $B_j$, $j\in\{1,2\}$, $w_{B_j} = r_{B_j}$.}
    \label{fig:model}
\end{figure}

\noindent{Each local plant $G_j=M_j||M_{j+1}$, $j=\{1,2\}$, is also a synchronizing automaton w.r.t. the initial state, $G_j=(\_,\Sigma_j\cup\Sigma_{j+1}\cup\{r_j,r_{j+1}\},\_,\_,\_)$, with shortest synchronization word $w_j\in\{{r_j}r_{j+1},\,r_{j+1}r_j\}$. For $G_1$ (Fig.~\ref{fig:ex}(a)), $w_1\in\{{r_1}r_{2},\,r_{2}r_1\}$ and for $G_2$ (Fig.~\ref{fig:exb2}(a)), $w_2\in\{{r_2}r_{3},\,r_{3}r_2\}$.}

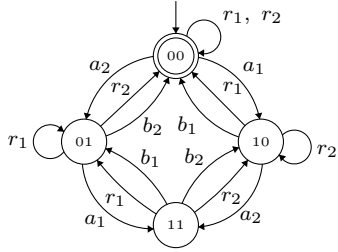
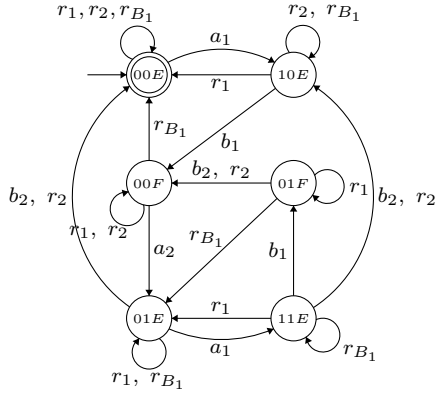
\begin{figure}[htbp]
    
    \centering
    \subfloat[$G_1=M_1||M_{2}$, with a (not unique) shortest synchronizing word $w_1= r_1  r_{2}$.]{
    \makebox[8cm][c]{
        \begin{tikzpicture}[scale=0.10] \footnotesize
          \tikzstyle{every node}+=[inner sep=0pt]
\draw [black] (33.4,-13.5) circle (3);
\draw (33.4,-13.5) node {\tiny$00$};
\draw [black] (33.4,-13.5) circle (2.4);
\draw [black] (33.4,-36) circle (3);
\draw (33.4,-36) node {\tiny$11$};
\draw [black] (44.7,-24.9) circle (3);
\draw (44.7,-24.9) node {\tiny$10$};
\draw [black] (21.2,-24.9) circle (3);
\draw (21.2,-24.9) node {\tiny$01$};
\draw [black] (33.4,-6.2) -- (33.4,-10.5);
\fill [black] (33.4,-10.5) -- (33.9,-9.7) -- (32.9,-9.7);
\draw [black] (42.59,-22.77) -- (35.51,-15.63);
\fill [black] (35.51,-15.63) -- (35.72,-16.55) -- (36.43,-15.85);
\draw (39.57,-17.72) node [right] {$r_1$};
\draw [black] (42.964,-27.346) arc (-37.71812:-53.30499:36.634);
\fill [black] (42.96,-27.35) -- (42.08,-27.67) -- (42.87,-28.28);
\draw (40.81,-31.55) node [below] {$r_{2}$};
\draw [black] (30.793,-34.517) arc (-122.40579:-142.18832:30.977);
\fill [black] (22.92,-27.36) -- (23.02,-28.29) -- (23.81,-27.68);
\draw (25.4,-31.76) node [below] {$r_1$};
\draw [black] (23.39,-22.85) -- (31.21,-15.55);
\fill [black] (31.21,-15.55) -- (30.28,-15.73) -- (30.96,-16.46);
\draw (26.15,-18.72) node [above] {$r_{2}$};
\draw [black] (18.52,-26.223) arc (-36:-324:2.25);
\draw (13.95,-24.9) node [left] {$r_1$};
\fill [black] (18.52,-23.58) -- (18.17,-22.7) -- (17.58,-23.51);
\draw [black] (47.556,-24.02) arc (134.85852:-153.14148:2.25);
\draw (52.07,-26.19) node [right] {$r_{2}$};
\fill [black] (47.14,-26.63) -- (47.35,-27.55) -- (48.05,-26.85);
\draw [black] (34.869,-10.898) arc (178.28688:-109.71312:2.25);
\draw (39.58,-8.35) node [right] {$r_1,\mbox{ }r_{2}$};
\fill [black] (36.36,-13.08) -- (37.14,-13.61) -- (37.17,-12.61);
\draw [black] (36.384,-13.678) arc (78.37519:11.12:10.47);
\fill [black] (44.55,-21.91) -- (44.88,-21.03) -- (43.9,-21.23);
\draw (42.24,-15.09) node [right] {$a_1$};
\draw [black] (41.772,-24.279) arc (-108.88739:-161.61741:12.438);
\fill [black] (34,-16.43) -- (33.77,-17.35) -- (34.72,-17.03);
\draw (36.44,-22.74) node [left] {$b_1$};
\draw [black] (44.902,-27.88) arc (-5.38471:-85.63841:9.272);
\fill [black] (36.38,-36.26) -- (37.21,-36.69) -- (37.14,-35.7);
\draw (43.43,-34.1) node [below] {$a_{2}$};
\draw [black] (34.224,-33.122) arc (157.66125:111.31564:13.507);
\fill [black] (41.81,-25.67) -- (40.88,-25.5) -- (41.24,-26.43);
\draw (35.93,-28.14) node [above] {$b_{2}$};
\draw [black] (24.054,-25.813) arc (67.0943:28.31159:16.632);
\fill [black] (24.05,-25.81) -- (24.6,-26.58) -- (24.99,-25.66);
\draw (30.09,-28.34) node [above] {$b_1$};
\draw [black] (30.456,-36.507) arc (-89.36058:-175.23354:9.411);
\fill [black] (30.46,-36.51) -- (29.65,-36.02) -- (29.66,-37.02);
\draw (22.76,-34.54) node [below] {$a_1$};
\draw [black] (32.455,-16.342) arc (-24.20196:-69.68102:14.793);
\fill [black] (32.46,-16.34) -- (31.67,-16.87) -- (32.58,-17.28);
\draw (30.38,-21.57) node [below] {$b_{2}$};
\draw [black] (21.464,-21.921) arc (-192.91597:-260.967:10.941);
\fill [black] (21.46,-21.92) -- (22.13,-21.25) -- (21.16,-21.03);
\draw (23.4,-15.89) node [above] {$a_{2}$};
        \end{tikzpicture}    \label{fig:comp}
       }
    }

    \subfloat[Resulting supervisor $S_1$, with a (not unique) shortest synchronizing word $w = r_1 r_{2} r_{B_1}$.]{
    \makebox[8cm][c]{

\begin{tikzpicture}[scale=0.10] \footnotesize
\tikzstyle{every node}+=[inner sep=0pt]
\draw [black] (14.9,-12.4) circle (3);
\draw (14.9,-12.4) node {\tiny $00E$};
\draw [black] (14.9,-12.4) circle (2.4);
\draw [black] (14.9,-26.8) circle (3);
\draw (14.9,-26.8) node {\tiny$00F$};
\draw [black] (34.1,-26.8) circle (3);
\draw (34.1,-26.8) node {\tiny$01F$};
\draw [black] (14.9,-44.8) circle (3);
\draw (14.9,-44.8) node {\tiny$01E$};
\draw [black] (34.1,-44.8) circle (3);
\draw (34.1,-44.8) node {\tiny$11E$};
\draw [black] (34.1,-12.4) circle (3);
\draw (34.1,-12.4) node {\tiny$10E$};
\draw [black] (6.7,-12.4) -- (11.9,-12.4);
\fill [black] (11.9,-12.4) -- (11.1,-11.9) -- (11.1,-12.9);
\draw [black] (12.74,-10.335) arc (254.02753:-33.97247:2.25);
\draw (9.25,-5.46) node [above] {$r_1,r_{2},r_{B_1}$};
\fill [black] (15.22,-9.43) -- (15.93,-8.8) -- (14.96,-8.52);
\draw [black] (17.4,-10.75) arc (117.78867:62.21133:15.23);
\fill [black] (31.6,-10.75) -- (31.13,-9.93) -- (30.66,-10.82);
\draw (24.5,-8.49) node [above] {$a_1$};
\draw [black] (31.1,-12.4) -- (17.9,-12.4);
\fill [black] (17.9,-12.4) -- (18.7,-12.9) -- (18.7,-11.9);
\draw (24.5,-12.9) node [below] {$r_1$};
\draw [black] (33.565,-9.46) arc (218.0546:-69.9454:2.25);
\draw (38.35,-5.25) node [above] {$r_{2},\mbox{ }r_{B_1}$};
\fill [black] (36.11,-10.19) -- (37.05,-10.09) -- (36.43,-9.3);
\draw [black] (36.737,-13.823) arc (56.78847:-56.78847:17.662);
\fill [black] (36.74,-13.82) -- (37.13,-14.68) -- (37.68,-13.84);
\draw (45.23,-28.6) node [right] {$b_{2},\mbox{ }r_{2}$};
\draw [black] (37.083,-44.981) arc (114.25512:-173.74488:2.25);
\draw (40.62,-49.37) node [right] {$r_{B_1}$};
\fill [black] (35.77,-47.28) -- (35.64,-48.21) -- (36.55,-47.8);
\draw [black] (31.7,-14.2) -- (17.3,-25);
\fill [black] (17.3,-25) -- (18.24,-24.92) -- (17.64,-24.12);
\draw (25.8,-20.1) node [below] {$b_1$};
\draw [black] (14.163,-29.696) arc (13.45497:-274.54503:2.25);
\draw (8.24,-32.42) node [below] {$r_1,\mbox{ }r_{2}$};
\fill [black] (12.15,-27.98) -- (11.26,-27.68) -- (11.49,-28.65);
\draw [black] (14.9,-23.8) -- (14.9,-15.4);
\fill [black] (14.9,-15.4) -- (14.4,-16.2) -- (15.4,-16.2);
\draw (15.4,-19.6) node [right] {$r_{B_1}$};
\draw [black] (31.1,-26.8) -- (17.9,-26.8);
\fill [black] (17.9,-26.8) -- (18.7,-27.3) -- (18.7,-26.3);
\draw (24.5,-26.3) node [above] {$b_{2},\mbox{ }r_{2}$};
\draw [black] (36.902,-25.762) arc (138.06524:-149.93476:2.25);
\draw (41.46,-27.59) node [right] {$r_1$};
\fill [black] (36.63,-28.39) -- (36.89,-29.3) -- (37.56,-28.56);
\draw [black] (31.91,-28.85) -- (17.09,-42.75);
\fill [black] (17.09,-42.75) -- (18.01,-42.57) -- (17.33,-41.84);
\draw (22.65,-35.32) node [above] {$r_{B_1}$};
\draw [black] (34.1,-41.8) -- (34.1,-29.8);
\fill [black] (34.1,-29.8) -- (33.6,-30.6) -- (34.6,-30.6);
\draw (33.6,-35.8) node [left] {$b_1$};
\draw [black] (16.223,-47.48) arc (54:-234:2.25);
\draw (14.9,-52.05) node [below] {$r_1,\mbox{ }r_{B_1}$};
\fill [black] (13.58,-47.48) -- (12.7,-47.83) -- (13.51,-48.42);
\draw [black] (31.453,-46.204) arc (-66.90405:-113.09595:17.725);
\fill [black] (31.45,-46.2) -- (30.52,-46.06) -- (30.91,-46.98);
\draw (24.5,-48.12) node [below] {$a_1$};
\draw [black] (31.1,-44.8) -- (17.9,-44.8);
\fill [black] (17.9,-44.8) -- (18.7,-45.3) -- (18.7,-44.3);
\draw (24.5,-44.3) node [above] {$r_1$};
\draw [black] (12.314,-43.285) arc (-125.14577:-234.85423:17.96);
\fill [black] (12.31,-13.91) -- (11.37,-13.97) -- (11.95,-14.78);
\draw (4.19,-28.6) node [left] {$b_{2},\mbox{ }r_{2}$};
\draw [black] (14.9,-29.8) -- (14.9,-41.8);
\fill [black] (14.9,-41.8) -- (15.4,-41) -- (14.4,-41);
\draw (15.4,-35.8) node [right] {$a_{2}$};
\end{tikzpicture} \label{fig:sup}
   }
}
\caption{Plant $G_1$ and supervisor $S_1$, obtained for the Extended Small Factory of Fig.~\ref{fig:sf}. }\label{fig:ex}
\end{figure}

\begin{figure}[htbp]
\centering
 \subfloat[$G_2$]{
        \begin{tikzpicture}[scale=0.1] \footnotesize
          \tikzstyle{every node}+=[inner sep=0pt]
\draw [black] (33.4,-13.5) circle (3);
\draw (33.4,-13.5) node {\tiny$00$};
\draw [black] (33.4,-13.5) circle (2.4);
\draw [black] (33.4,-36) circle (3);
\draw (33.4,-36) node {\tiny$11$};
\draw [black] (44.7,-24.9) circle (3);
\draw (44.7,-24.9) node {\tiny$10$};
\draw [black] (21.2,-24.9) circle (3);
\draw (21.2,-24.9) node {\tiny$01$};
\draw [black] (33.4,-6.2) -- (33.4,-10.5);
\fill [black] (33.4,-10.5) -- (33.9,-9.7) -- (32.9,-9.7);
\draw [black] (42.59,-22.77) -- (35.51,-15.63);
\fill [black] (35.51,-15.63) -- (35.72,-16.55) -- (36.43,-15.85);
\draw (39.57,-17.72) node [right] {$r_2$};
\draw [black] (42.964,-27.346) arc (-37.71812:-53.30499:36.634);
\fill [black] (42.96,-27.35) -- (42.08,-27.67) -- (42.87,-28.28);
\draw (40.81,-31.55) node [below] {$r_{3}$};
\draw [black] (30.793,-34.517) arc (-122.40579:-142.18832:30.977);
\fill [black] (22.92,-27.36) -- (23.02,-28.29) -- (23.81,-27.68);
\draw (25.4,-31.76) node [below] {$r_2$};
\draw [black] (23.39,-22.85) -- (31.21,-15.55);
\fill [black] (31.21,-15.55) -- (30.28,-15.73) -- (30.96,-16.46);
\draw (26.15,-18.72) node [above] {$r_{3}$};
\draw [black] (18.52,-26.223) arc (-36:-324:2.25);
\draw (13.95,-24.9) node [left] {$r_2$};
\fill [black] (18.52,-23.58) -- (18.17,-22.7) -- (17.58,-23.51);
\draw [black] (47.556,-24.02) arc (134.85852:-153.14148:2.25);
\draw (52.07,-26.19) node [right] {$r_{3}$};
\fill [black] (47.14,-26.63) -- (47.35,-27.55) -- (48.05,-26.85);
\draw [black] (34.869,-10.898) arc (178.28688:-109.71312:2.25);
\draw (39.58,-8.35) node [right] {$r_2,\mbox{ }r_{3}$};
\fill [black] (36.36,-13.08) -- (37.14,-13.61) -- (37.17,-12.61);
\draw [black] (36.384,-13.678) arc (78.37519:11.12:10.47);
\fill [black] (44.55,-21.91) -- (44.88,-21.03) -- (43.9,-21.23);
\draw (42.24,-15.09) node [right] {$a_2$};
\draw [black] (41.772,-24.279) arc (-108.88739:-161.61741:12.438);
\fill [black] (34,-16.43) -- (33.77,-17.35) -- (34.72,-17.03);
\draw (36.44,-22.74) node [left] {$b_2$};
\draw [black] (44.902,-27.88) arc (-5.38471:-85.63841:9.272);
\fill [black] (36.38,-36.26) -- (37.21,-36.69) -- (37.14,-35.7);
\draw (43.43,-34.1) node [below] {$a_{3}$};
\draw [black] (34.224,-33.122) arc (157.66125:111.31564:13.507);
\fill [black] (41.81,-25.67) -- (40.88,-25.5) -- (41.24,-26.43);
\draw (35.93,-28.14) node [above] {$b_{3}$};
\draw [black] (24.054,-25.813) arc (67.0943:28.31159:16.632);
\fill [black] (24.05,-25.81) -- (24.6,-26.58) -- (24.99,-25.66);
\draw (30.09,-28.34) node [above] {$b_2$};
\draw [black] (30.456,-36.507) arc (-89.36058:-175.23354:9.411);
\fill [black] (30.46,-36.51) -- (29.65,-36.02) -- (29.66,-37.02);
\draw (22.76,-34.54) node [below] {$a_2$};
\draw [black] (32.455,-16.342) arc (-24.20196:-69.68102:14.793);
\fill [black] (32.46,-16.34) -- (31.67,-16.87) -- (32.58,-17.28);
\draw (30.38,-21.57) node [below] {$b_{3}$};
\draw [black] (21.464,-21.921) arc (-192.91597:-260.967:10.941);
\fill [black] (21.46,-21.92) -- (22.13,-21.25) -- (21.16,-21.03);
\draw (23.4,-15.89) node [above] {$a_{3}$};
        \end{tikzpicture}   
       }\\
\subfloat[$S_2$]{
\begin{tikzpicture}[scale=0.1] \footnotesize
\tikzstyle{every node}+=[inner sep=0pt]
\draw [black] (14.9,-12.4) circle (3);
\draw (14.9,-12.4) node {\tiny $00E$};
\draw [black] (14.9,-12.4) circle (2.4);
\draw [black] (14.9,-26.8) circle (3);
\draw (14.9,-26.8) node {\tiny$00F$};
\draw [black] (34.1,-26.8) circle (3);
\draw (34.1,-26.8) node {\tiny$01F$};
\draw [black] (14.9,-44.8) circle (3);
\draw (14.9,-44.8) node {\tiny$01E$};
\draw [black] (34.1,-44.8) circle (3);
\draw (34.1,-44.8) node {\tiny$11E$};
\draw [black] (34.1,-12.4) circle (3);
\draw (34.1,-12.4) node {\tiny$10E$};
\draw [black] (6.7,-12.4) -- (11.9,-12.4);
\fill [black] (11.9,-12.4) -- (11.1,-11.9) -- (11.1,-12.9);
\draw [black] (12.74,-10.335) arc (254.02753:-33.97247:2.25);
\draw (9.25,-5.46) node [above] {$r_2,r_{3},r_{B_2}$};
\fill [black] (15.22,-9.43) -- (15.93,-8.8) -- (14.96,-8.52);
\draw [black] (17.4,-10.75) arc (117.78867:62.21133:15.23);
\fill [black] (31.6,-10.75) -- (31.13,-9.93) -- (30.66,-10.82);
\draw (24.5,-8.49) node [above] {$a_2$};
\draw [black] (31.1,-12.4) -- (17.9,-12.4);
\fill [black] (17.9,-12.4) -- (18.7,-12.9) -- (18.7,-11.9);
\draw (24.5,-12.9) node [below] {$r_2$};
\draw [black] (33.565,-9.46) arc (218.0546:-69.9454:2.25);
\draw (38.35,-5.25) node [above] {$r_{3},\mbox{ }r_{B_2}$};
\fill [black] (36.11,-10.19) -- (37.05,-10.09) -- (36.43,-9.3);
\draw [black] (36.737,-13.823) arc (56.78847:-56.78847:17.662);
\fill [black] (36.74,-13.82) -- (37.13,-14.68) -- (37.68,-13.84);
\draw (45.23,-28.6) node [right] {$b_{3},\mbox{ }r_{3}$};
\draw [black] (37.083,-44.981) arc (114.25512:-173.74488:2.25);
\draw (40.62,-49.37) node [right] {$r_{B_2}$};
\fill [black] (35.77,-47.28) -- (35.64,-48.21) -- (36.55,-47.8);
\draw [black] (31.7,-14.2) -- (17.3,-25);
\fill [black] (17.3,-25) -- (18.24,-24.92) -- (17.64,-24.12);
\draw (25.8,-20.1) node [below] {$b_2$};
\draw [black] (14.163,-29.696) arc (13.45497:-274.54503:2.25);
\draw (8.24,-32.42) node [below] {$r_2,\mbox{ }r_{3}$};
\fill [black] (12.15,-27.98) -- (11.26,-27.68) -- (11.49,-28.65);
\draw [black] (14.9,-23.8) -- (14.9,-15.4);
\fill [black] (14.9,-15.4) -- (14.4,-16.2) -- (15.4,-16.2);
\draw (15.4,-19.6) node [right] {$r_{B_2}$};
\draw [black] (31.1,-26.8) -- (17.9,-26.8);
\fill [black] (17.9,-26.8) -- (18.7,-27.3) -- (18.7,-26.3);
\draw (24.5,-26.3) node [above] {$b_{3},\mbox{ }r_{3}$};
\draw [black] (36.902,-25.762) arc (138.06524:-149.93476:2.25);
\draw (41.46,-27.59) node [right] {$r_2$};
\fill [black] (36.63,-28.39) -- (36.89,-29.3) -- (37.56,-28.56);
\draw [black] (31.91,-28.85) -- (17.09,-42.75);
\fill [black] (17.09,-42.75) -- (18.01,-42.57) -- (17.33,-41.84);
\draw (22.65,-35.32) node [above] {$r_{B_2}$};
\draw [black] (34.1,-41.8) -- (34.1,-29.8);
\fill [black] (34.1,-29.8) -- (33.6,-30.6) -- (34.6,-30.6);
\draw (33.6,-35.8) node [left] {$b_2$};
\draw [black] (16.223,-47.48) arc (54:-234:2.25);
\draw (14.9,-52.05) node [below] {$r_2,\mbox{ }r_{B_2}$};
\fill [black] (13.58,-47.48) -- (12.7,-47.83) -- (13.51,-48.42);
\draw [black] (31.453,-46.204) arc (-66.90405:-113.09595:17.725);
\fill [black] (31.45,-46.2) -- (30.52,-46.06) -- (30.91,-46.98);
\draw (24.5,-48.12) node [below] {$a_2$};
\draw [black] (31.1,-44.8) -- (17.9,-44.8);
\fill [black] (17.9,-44.8) -- (18.7,-45.3) -- (18.7,-44.3);
\draw (24.5,-44.3) node [above] {$r_2$};
\draw [black] (12.314,-43.285) arc (-125.14577:-234.85423:17.96);
\fill [black] (12.31,-13.91) -- (11.37,-13.97) -- (11.95,-14.78);
\draw (4.19,-28.6) node [left] {$b_{3},\mbox{ }r_{3}$};
\draw [black] (14.9,-29.8) -- (14.9,-41.8);
\fill [black] (14.9,-41.8) -- (15.4,-41) -- (14.4,-41);
\draw (15.4,-35.8) node [right] {$a_{3}$};
\end{tikzpicture} 
}
\caption{Plant $G_2=M_2||M_3$ with synchronizing word $w= r_2  r_{3}$ and supervisor $S_2$ with synchronizing word $w = r_2 r_{3} r_{B_2}$}\label{fig:exb2}
\end{figure}
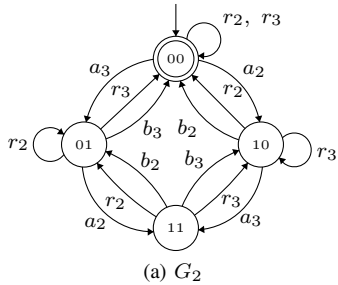
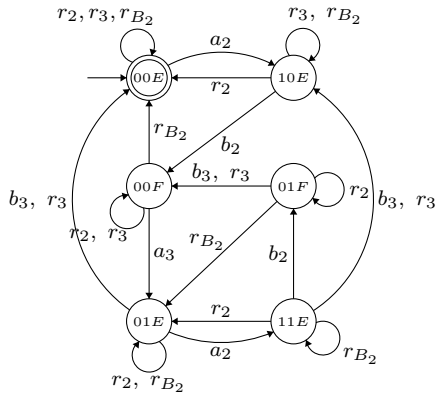
\label{ex:extended1} 

We propose to execute a synchronizing word of the supervisor, that includes the recovery events of the plant, in order to recover the system. It is straightforward that any word that brings each original automaton, inverse projected to the same event set (union of all sets), to their initial states, leads the composition to its initial state.
So, consider a synchronization word of $S_1$, $w_1=r_1r_2r_{B_1}$. If we execute $w_1$ in $S_1$, Fig.\ref{fig:ex}(b), regardless the original state, we will reach the initial state. Moreover, any synchronizing word, built for $S_1$ resets also its correspondent local plant and buffer. If we run $w_1$ in $G_1$ (Fig.\ref{fig:ex}(a)), and its corresponding natural projections in $M_1$, $M_2$ and $B_1$ (Fig.\ref{fig:model}(a) and (b)) it will lead us to the initial state.

\begin{figure}[b]
    \centering
    \input{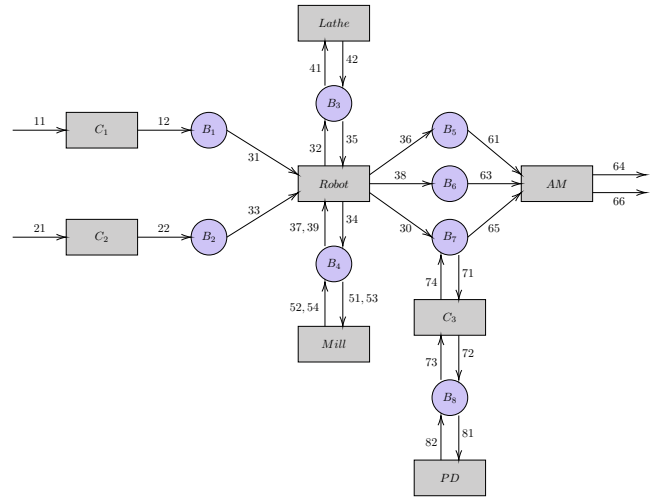}
    \caption{Diagram of the Flexible Manufacturing System}
    \label{fig:SFMdiag}
\end{figure}

\begin{figure*}[t]
    \centering
    \subfloat[$C_1$]{ \noindent\scalebox{0.55}{
        \begin{tikzpicture}[scale=0.2] 
            \tikzstyle{every node}+=[inner sep=0pt]
            \draw [black] (10.9,-13.8) circle (3);
            \draw (10.9,-13.8) node {$0$};
            \draw [black] (10.9,-13.8) circle (2.4);
            \draw [black] (28.6,-13.8) circle (3);
            \draw (28.6,-13.8) node {$1$};
            \draw [black] (13.295,-12.005) arc (120.15825:59.84175:12.848);
            \fill [black] (26.2,-12.01) -- (25.76,-11.17) -- (25.26,-12.04);
            \draw (19.75,-9.77) node [above] {$11$};
            \draw [black] (26.105,-15.455) arc (-62.66259:-117.33741:13.838);
            \fill [black] (13.4,-15.45) -- (13.88,-16.27) -- (14.34,-15.38);
            \draw (19.75,-17.5) node [below] {$12,r_1$};
            \draw [black] (8.404,-12.157) arc (264.38332:-23.61668:2.25);
            \draw (6.35,-7.48) node [above] {$r_1$};
            \fill [black] (10.69,-10.82) -- (11.26,-10.07) -- (10.27,-9.97);
            \draw [black] (3.8,-13.8) -- (7.9,-13.8);
            \fill [black] (7.9,-13.8) -- (7.1,-13.3) -- (7.1,-14.3);
        \end{tikzpicture} }
    }
    \subfloat[$C_2$]{ \noindent\scalebox{0.55}{
        \begin{tikzpicture}[scale=0.2] 
            \tikzstyle{every node}+=[inner sep=0pt]
            \draw [black] (10.9,-13.8) circle (3);
            \draw (10.9,-13.8) node {$0$};
            \draw [black] (10.9,-13.8) circle (2.4);
            \draw [black] (28.6,-13.8) circle (3);
            \draw (28.6,-13.8) node {$1$};
            \draw [black] (13.295,-12.005) arc (120.15825:59.84175:12.848);
            \fill [black] (26.2,-12.01) -- (25.76,-11.17) -- (25.26,-12.04);
            \draw (19.75,-9.77) node [above] {$21$};
            \draw [black] (26.105,-15.455) arc (-62.66259:-117.33741:13.838);
            \fill [black] (13.4,-15.45) -- (13.88,-16.27) -- (14.34,-15.38);
            \draw (19.75,-17.5) node [below] {$22,r_2$};
            \draw [black] (8.404,-12.157) arc (264.38332:-23.61668:2.25);
            \draw (6.35,-7.48) node [above] {$r_2$};
            \fill [black] (10.69,-10.82) -- (11.26,-10.07) -- (10.27,-9.97);
            \draw [black] (3.8,-13.8) -- (7.9,-13.8);
            \fill [black] (7.9,-13.8) -- (7.1,-13.3) -- (7.1,-14.3);
        \end{tikzpicture} }
    }
    \subfloat[$Lathe$]{ \noindent\scalebox{0.55}{
        \begin{tikzpicture}[scale=0.2] 
            \tikzstyle{every node}+=[inner sep=0pt]
            \draw [black] (10.9,-13.8) circle (3);
            \draw (10.9,-13.8) node {$0$};
            \draw [black] (10.9,-13.8) circle (2.4);
            \draw [black] (28.6,-13.8) circle (3);
            \draw (28.6,-13.8) node {$1$};
            \draw [black] (13.295,-12.005) arc (120.15825:59.84175:12.848);
            \fill [black] (26.2,-12.01) -- (25.76,-11.17) -- (25.26,-12.04);
            \draw (19.75,-9.77) node [above] {$41$};
            \draw [black] (26.105,-15.455) arc (-62.66259:-117.33741:13.838);
            \fill [black] (13.4,-15.45) -- (13.88,-16.27) -- (14.34,-15.38);
            \draw (19.75,-17.5) node [below] {$42,r_l$};
            \draw [black] (8.404,-12.157) arc (264.38332:-23.61668:2.25);
            \draw (6.35,-7.48) node [above] {$r_l$};
            \fill [black] (10.69,-10.82) -- (11.26,-10.07) -- (10.27,-9.97);
            \draw [black] (3.8,-13.8) -- (7.9,-13.8);
            \fill [black] (7.9,-13.8) -- (7.1,-13.3) -- (7.1,-14.3);
        \end{tikzpicture} }
    }
    \subfloat[$PD$]{ \noindent\scalebox{0.55}{
        \begin{tikzpicture}[scale=0.2] 
            \tikzstyle{every node}+=[inner sep=0pt]
            \draw [black] (10.9,-13.8) circle (3);
            \draw (10.9,-13.8) node {$0$};
            \draw [black] (10.9,-13.8) circle (2.4);
            \draw [black] (28.6,-13.8) circle (3);
            \draw (28.6,-13.8) node {$1$};
            \draw [black] (13.295,-12.005) arc (120.15825:59.84175:12.848);
            \fill [black] (26.2,-12.01) -- (25.76,-11.17) -- (25.26,-12.04);
            \draw (19.75,-9.77) node [above] {$81$};
            \draw [black] (26.105,-15.455) arc (-62.66259:-117.33741:13.838);
            \fill [black] (13.4,-15.45) -- (13.88,-16.27) -- (14.34,-15.38);
            \draw (19.75,-17.5) node [below] {$82,r_p$};
            \draw [black] (8.404,-12.157) arc (264.38332:-23.61668:2.25);
            \draw (6.35,-7.48) node [above] {$r_p$};
            \fill [black] (10.69,-10.82) -- (11.26,-10.07) -- (10.27,-9.97);
            \draw [black] (3.8,-13.8) -- (7.9,-13.8);
            \fill [black] (7.9,-13.8) -- (7.1,-13.3) -- (7.1,-14.3);
        \end{tikzpicture} }
    }\\
    \subfloat[$Mill$]{ \noindent\scalebox{0.55}{
        \begin{tikzpicture}[scale=0.2] 
            \tikzstyle{every node}+=[inner sep=0pt]
            \draw [black] (10.9,-13.8) circle (3);
            \draw (10.9,-13.8) node {$0$};
            \draw [black] (10.9,-13.8) circle (2.4);
            \draw [black] (28.6,-13.8) circle (3);
            \draw (28.6,-13.8) node {$1$};
            \draw [black] (10.9,-28.7) circle (3);
            \draw (10.9,-28.7) node {$2$};
            \draw [black] (13.498,-12.309) arc (114.21463:65.78537:15.244);
            \fill [black] (26,-12.31) -- (25.48,-11.52) -- (25.07,-12.44);
            \draw (19.75,-10.47) node [above] {$51$};
            \draw [black] (25.847,-14.984) arc (-71.25615:-108.74385:18.974);
            \fill [black] (13.65,-14.98) -- (14.25,-15.71) -- (14.57,-14.77);
            \draw (19.75,-16.49) node [below] {$52,r_m$};
            \draw [black] (8.404,-12.157) arc (264.38332:-23.61668:2.25);
            \draw (5.88,-7.48) node [above] {$r_m$};
            \fill [black] (10.69,-10.82) -- (11.26,-10.07) -- (10.27,-9.97);
            \draw [black] (3.8,-13.8) -- (7.9,-13.8);
            \fill [black] (7.9,-13.8) -- (7.1,-13.3) -- (7.1,-14.3);
            \draw [black] (8.941,-26.444) arc (-147.83654:-212.16346:9.757);
            \fill [black] (8.94,-26.44) -- (8.94,-25.5) -- (8.09,-26.03);
            \draw (6.94,-21.25) node [left] {$53$};
            \draw [black] (12.31,-16.441) arc (21.53938:-21.53938:13.1);
            \fill [black] (12.31,-16.44) -- (12.14,-17.37) -- (13.07,-17);
            \draw (13.72,-21.25) node [right] {$54,r_m$};
        \end{tikzpicture} }
    }
    \subfloat[$C_3$]{ \noindent\scalebox{0.55}{
        \begin{tikzpicture}[scale=0.2] 
            \tikzstyle{every node}+=[inner sep=0pt]
            \draw [black] (10.9,-13.8) circle (3);
            \draw (10.9,-13.8) node {$0$};
            \draw [black] (10.9,-13.8) circle (2.4);
            \draw [black] (28.6,-13.8) circle (3);
            \draw (28.6,-13.8) node {$1$};
            \draw [black] (10.9,-28.7) circle (3);
            \draw (10.9,-28.7) node {$2$};
            \draw [black] (13.498,-12.309) arc (114.21463:65.78537:15.244);
            \fill [black] (26,-12.31) -- (25.48,-11.52) -- (25.07,-12.44);
            \draw (19.75,-10.47) node [above] {$71$};
            \draw [black] (25.847,-14.984) arc (-71.25615:-108.74385:18.974);
            \fill [black] (13.65,-14.98) -- (14.25,-15.71) -- (14.57,-14.77);
            \draw (19.75,-16.49) node [below] {$72,r_3$};
            \draw [black] (8.404,-12.157) arc (264.38332:-23.61668:2.25);
            \draw (5.88,-7.48) node [above] {$r_3$};
            \fill [black] (10.69,-10.82) -- (11.26,-10.07) -- (10.27,-9.97);
            \draw [black] (3.8,-13.8) -- (7.9,-13.8);
            \fill [black] (7.9,-13.8) -- (7.1,-13.3) -- (7.1,-14.3);
            \draw [black] (8.941,-26.444) arc (-147.83654:-212.16346:9.757);
            \fill [black] (8.94,-26.44) -- (8.94,-25.5) -- (8.09,-26.03);
            \draw (6.94,-21.25) node [left] {$73$};
            \draw [black] (12.31,-16.441) arc (21.53938:-21.53938:13.1);
            \fill [black] (12.31,-16.44) -- (12.14,-17.37) -- (13.07,-17);
            \draw (13.72,-21.25) node [right] {$74,r_3$};
        \end{tikzpicture} }
    }
    \subfloat[$AM$]{ \noindent\scalebox{0.55}{
        \begin{tikzpicture}[scale=0.2] 
            \tikzstyle{every node}+=[inner sep=0pt]
            \draw [black] (11.1,-21.4) circle (3);
            \draw (11.1,-21.4) node {$0$};
            \draw [black] (11.1,-21.4) circle (2.4);
            \draw [black] (27.1,-21.4) circle (3);
            \draw (27.1,-21.4) node {$1$};
            \draw [black] (36.5,-13.2) circle (3);
            \draw (36.5,-13.2) node {$2$};
            \draw [black] (36.5,-28.7) circle (3);
            \draw (36.5,-28.7) node {$3$};
            \draw [black] (8.536,-19.865) arc (266.82939:-21.17061:2.25);
            \draw (5.33,-15.23) node [above] {$r_a$};
            \fill [black] (10.76,-18.43) -- (11.3,-17.66) -- (10.3,-17.6);
            \draw [black] (4,-21.4) -- (8.1,-21.4);
            \fill [black] (8.1,-21.4) -- (7.3,-20.9) -- (7.3,-21.9);
            \draw [black] (13.584,-19.731) arc (116.86027:63.13973:12.209);
            \fill [black] (24.62,-19.73) -- (24.13,-18.92) -- (23.68,-19.82);
            \draw (19.1,-17.91) node [above] {$61$};
            \draw [black] (24.374,-22.643) arc (-70.83251:-109.16749:16.064);
            \fill [black] (13.83,-22.64) -- (14.42,-23.38) -- (14.75,-22.43);
            \draw (19.1,-24.03) node [below] {$r_a$};
            \draw [black] (29.36,-19.43) -- (34.24,-15.17);
            \fill [black] (34.24,-15.17) -- (33.31,-15.32) -- (33.97,-16.07);
            \draw (33.31,-17.79) node [below] {$63$};
            \draw [black] (12.582,-18.795) arc (145.6361:70.14766:18.19);
            \fill [black] (12.58,-18.8) -- (13.45,-18.42) -- (12.62,-17.85);
            \draw (18.85,-11.11) node [above] {$64,r_a$};
            \draw [black] (29.47,-23.24) -- (34.13,-26.86);
            \fill [black] (34.13,-26.86) -- (33.81,-25.97) -- (33.19,-26.76);
            \draw (30.29,-25.55) node [below] {$65$};
            \draw [black] (33.844,-30.087) arc (-67.30455:-144.76492:17.652);
            \fill [black] (12.61,-23.99) -- (12.67,-24.93) -- (13.48,-24.35);
            \draw (19.23,-31.46) node [below] {$66,r_a$};
        \end{tikzpicture} }
    }
    \subfloat[$Robot$]{ \noindent\scalebox{0.55}{
        \begin{tikzpicture}[scale=0.2] 
            \tikzstyle{every node}+=[inner sep=0pt]
            \draw [black] (36.4,-14.6) circle (3);
            \draw (36.4,-14.6) node {$0$};
            \draw [black] (36.4,-14.6) circle (2.4);
            \draw [black] (36.4,-29.8) circle (3);
            \draw (36.4,-29.8) node {$3$};
            \draw [black] (52.5,-29.8) circle (3);
            \draw (52.5,-29.8) node {$4$};
            \draw [black] (20.2,-29.8) circle (3);
            \draw (20.2,-29.8) node {$2$};
            \draw [black] (16.7,-14.6) circle (3);
            \draw (16.7,-14.6) node {$1$};
            \draw [black] (55,-14.6) circle (3);
            \draw (55,-14.6) node {$5$};
            \draw [black] (35.077,-11.92) arc (234:-54:2.25);
            \draw (36.4,-7.35) node [above] {$r_r$};
            \fill [black] (37.72,-11.92) -- (38.6,-11.57) -- (37.79,-10.98);
            \draw [black] (32.6,-8.9) -- (34.74,-12.1);
            \fill [black] (34.74,-12.1) -- (34.71,-11.16) -- (33.88,-11.72);
            \draw [black] (37.251,-17.474) arc (12.5487:-12.5487:21.751);
            \fill [black] (37.25,-26.93) -- (37.91,-26.25) -- (36.94,-26.04);
            \draw (38.27,-22.2) node [right] {$35$};
            \draw [black] (35.645,-26.898) arc (-168.93002:-191.06998:24.47);
            \fill [black] (35.65,-17.5) -- (35,-18.19) -- (35.98,-18.38);
            \draw (34.69,-22.2) node [left] {$36,r_r$};
            \draw [black] (38.58,-16.66) -- (50.32,-27.74);
            \fill [black] (50.32,-27.74) -- (50.08,-26.83) -- (49.39,-27.55);
            \draw (42.93,-22.68) node [below] {$37$};
            \draw [black] (18.522,-12.227) arc (134.92498:45.07502:11.369);
            \fill [black] (18.52,-12.23) -- (19.44,-12.02) -- (18.73,-11.31);
            \draw (26.55,-8.41) node [above] {$31$};
            \draw [black] (37.887,-12.008) arc (141.54586:38.45414:9.977);
            \fill [black] (53.51,-12.01) -- (53.41,-11.07) -- (52.62,-11.69);
            \draw (45.7,-7.74) node [above] {$39$};
            \draw [black] (34.21,-16.65) -- (22.39,-27.75);
            \fill [black] (22.39,-27.75) -- (23.31,-27.56) -- (22.63,-26.84);
            \draw (26.78,-21.72) node [above] {$33$};
            \draw [black] (39.363,-15.047) arc (76.46596:16.82809:17.315);
            \fill [black] (39.36,-15.05) -- (40.02,-15.72) -- (40.26,-14.75);
            \draw (50.38,-18.81) node [above] {$38,r_r$};
            \draw [black] (20.547,-26.825) arc (167.80914:98.54268:15.517);
            \fill [black] (33.41,-14.76) -- (32.54,-14.38) -- (32.69,-15.37);
            \draw (22.1,-18.3) node [above] {$34,r_r$};
            \draw [black] (19.532,-13.616) arc (105.81689:74.18311:25.747);
            \fill [black] (33.57,-13.62) -- (32.93,-12.92) -- (32.66,-13.88);
            \draw (26.55,-12.14) node [above] {$32,r_r$};
            \draw [black] (39.015,-13.138) arc (113.97768:66.02232:16.45);
            \fill [black] (39.02,-13.14) -- (39.95,-13.27) -- (39.54,-12.36);
            \draw (45.7,-11.22) node [above] {$30,r_r$};
        \end{tikzpicture} }
    }
    \caption{Models of the plants of the Flexible Manufacturing System}
    \label{fig:SFM_plantas}
\end{figure*}
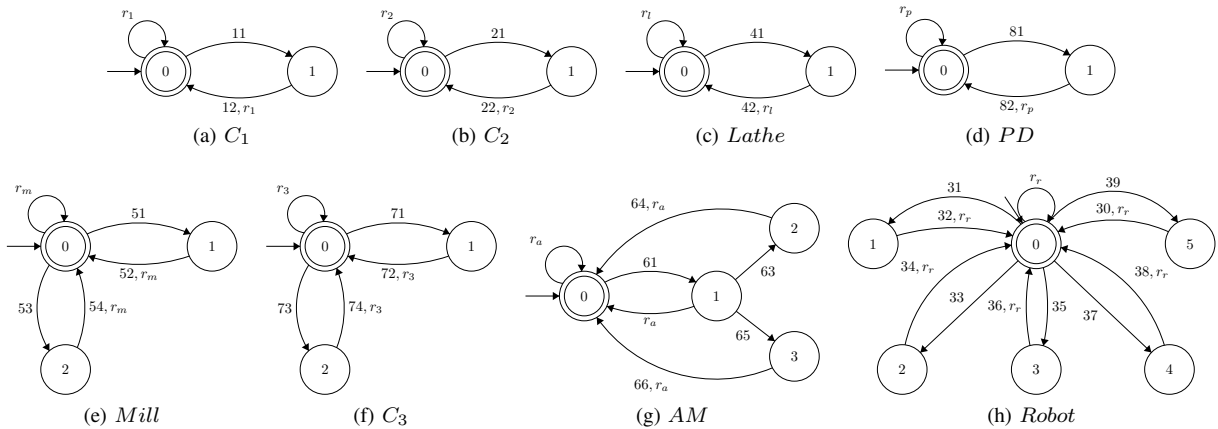

Now, we show how the desynchronization may happen and how to use the synchronizing word to solve it. If a sequence $s = a_1 \, b_1 \, a_2 \, a_1 \, b_2 a_3 \in\overline{S_1}||\overline{S_2}$ is executed, 
states (10) of $G_1$ and (10E) in $S_1$ are reached. Consider now that event $b_1$ happens in the plant but the supervisor does not observe it (a malicious agent has hidden such occurrence from the control). Automaton $G_1$ would move to state (00), and $S_1$ would stay at (10E), since $b_1$ is not observed. At this point, the system and control are desynchronized. Supervisor $S_1$ disables $a_2$ in state (10E) until $b_1$ is observed and $a_1$ is considered by the supervisor as not possible and the control systems reaches a deadlock. If we apply a synchronizing word of $S_1$, $w=r_1r_2r_{B_1}$,  $S_1$ and $G_1$ will move to the initial state.


Next, we analyze the effect that resynchronizing $S_1$ and $G_1$ causes in $S_2$ and $G_2$. The word $w=r_1r_2r_{B_1}$, resets $M_2$ that is a system that is shared by $G_1$ and $G_2$. Supervisor $S_2$ and plant $G_2$ will see $w'=r_2$ and will be kept in the same state (selfloops with $r_2$ in states (01) of $G_2$ and (01E) of $S_2$), namely the recovery is localized for $S_1$ and $G_1$ and $S_2$ and $G_2$ are kept as before.

In the following, we present a case study of a Flexible Manufacturing System (FMS), that illustrates the application of the security recovery procedure in a larger system.

\subsection{Flexible Manufacturing Systems}\label{sec:app}

The Flexible Manufacturing System (FMS) \cite{Queiroz2005} is composed of eight machines: three conveyors ($C_1$, $C_2$ and $C_3$), a mill, a lathe, a robot, a painting device (PD) and an assembly machine (AM), as shown on Figure~\ref{fig:SFMdiag}.

The automata for the subsystems, modeled as synchronizing automata w.r.t. the initial state with reset event, are shown in Figure~\ref{fig:SFM_plantas}. The safety specifications, that restrict the system to avoid underflow and overflow in the buffers, are presented in Figure~\ref{fig:SFM_espec}.

A total of 15 recovery events were created, one event for each plant and specification. Monolithic and  Local Modular Supervisory Control were applied in order to obtain a set of controllable, nonblocking and nonconflicting supervisors. The synthesis of supervisors was done using software \textit{UltraDES} \cite{alves2017ultrades}. The classical algorithms were adapted to cope with conditions of Corollary~\ref{cor:sct} and Definition~\ref{def:nonconflictReset}. 

The application of the monolithic approach leads to a single supervisor with $70,272$ states, $1,434,804$ transitions, being $1,054,080$ of these transitions triggered by reset events. The shortest synchronizing word has $16$ events and, if used,  will reset the system (all equipment and supervisors) to the initial state.

Using the same automata of figures \ref{fig:SFM_plantas} and \ref{fig:SFM_espec}, the resulting supervisors are conflicting (as they were in the solution without reset events). The conflict solution is to compose specifications $ E_7 $ and $ E_8 $ as a single local specification, generating 7 nonconflicting supervisors, as presented in Table~\ref{tab:sups}.

\begin{figure*}[t]
    \centering
      \subfloat[$E_1$]{ \noindent\scalebox{0.55}{
        \begin{tikzpicture}[scale=0.2] 
            \tikzstyle{every node}+=[inner sep=0pt]
            \draw [black] (10.9,-13.8) circle (3);
            \draw (10.9,-13.8) node {$0$};
            \draw [black] (10.9,-13.8) circle (2.4);
            \draw [black] (28.6,-13.8) circle (3);
            \draw (28.6,-13.8) node {$1$};
            \draw [black] (13.295,-12.005) arc (120.15825:59.84175:12.848);
            \fill [black] (26.2,-12.01) -- (25.76,-11.17) -- (25.26,-12.04);
            \draw (19.75,-9.77) node [above] {$12$};
            \draw [black] (26.105,-15.455) arc (-62.66259:-117.33741:13.838);
            \fill [black] (13.4,-15.45) -- (13.88,-16.27) -- (14.34,-15.38);
            \draw (19.75,-17.5) node [below] {$31,r_{B1}$};
            \draw [black] (8.404,-12.157) arc (264.38332:-23.61668:2.25);
            \draw (6.35,-7.48) node [above] {$r_{B1}$};
            \fill [black] (10.69,-10.82) -- (11.26,-10.07) -- (10.27,-9.97);
            \draw [black] (3.8,-13.8) -- (7.9,-13.8);
            \fill [black] (7.9,-13.8) -- (7.1,-13.3) -- (7.1,-14.3);
        \end{tikzpicture} }
    }
    \subfloat[$E_2$]{ \noindent\scalebox{0.55}{
        \begin{tikzpicture}[scale=0.2] 
            \tikzstyle{every node}+=[inner sep=0pt]
            \draw [black] (10.9,-13.8) circle (3);
            \draw (10.9,-13.8) node {$0$};
            \draw [black] (10.9,-13.8) circle (2.4);
            \draw [black] (28.6,-13.8) circle (3);
            \draw (28.6,-13.8) node {$1$};
            \draw [black] (13.295,-12.005) arc (120.15825:59.84175:12.848);
            \fill [black] (26.2,-12.01) -- (25.76,-11.17) -- (25.26,-12.04);
            \draw (19.75,-9.77) node [above] {$22$};
            \draw [black] (26.105,-15.455) arc (-62.66259:-117.33741:13.838);
            \fill [black] (13.4,-15.45) -- (13.88,-16.27) -- (14.34,-15.38);
            \draw (19.75,-17.5) node [below] {$33,r_{B2}$};
            \draw [black] (8.404,-12.157) arc (264.38332:-23.61668:2.25);
            \draw (6.35,-7.48) node [above] {$r_{B2}$};
            \fill [black] (10.69,-10.82) -- (11.26,-10.07) -- (10.27,-9.97);
            \draw [black] (3.8,-13.8) -- (7.9,-13.8);
            \fill [black] (7.9,-13.8) -- (7.1,-13.3) -- (7.1,-14.3);
        \end{tikzpicture} }
    }
    \subfloat[$E_5$]{ \noindent\scalebox{0.55}{
        \begin{tikzpicture}[scale=0.2] 
            \tikzstyle{every node}+=[inner sep=0pt]
            \draw [black] (10.9,-13.8) circle (3);
            \draw (10.9,-13.8) node {$0$};
            \draw [black] (10.9,-13.8) circle (2.4);
            \draw [black] (28.6,-13.8) circle (3);
            \draw (28.6,-13.8) node {$1$};
            \draw [black] (13.295,-12.005) arc (120.15825:59.84175:12.848);
            \fill [black] (26.2,-12.01) -- (25.76,-11.17) -- (25.26,-12.04);
            \draw (19.75,-9.77) node [above] {$36$};
            \draw [black] (26.105,-15.455) arc (-62.66259:-117.33741:13.838);
            \fill [black] (13.4,-15.45) -- (13.88,-16.27) -- (14.34,-15.38);
            \draw (19.75,-17.5) node [below] {$61,r_{B5}$};
            \draw [black] (8.404,-12.157) arc (264.38332:-23.61668:2.25);
            \draw (6.35,-7.48) node [above] {$r_{B5}$};
            \fill [black] (10.69,-10.82) -- (11.26,-10.07) -- (10.27,-9.97);
            \draw [black] (3.8,-13.8) -- (7.9,-13.8);
            \fill [black] (7.9,-13.8) -- (7.1,-13.3) -- (7.1,-14.3);
        \end{tikzpicture} }
    }
    \subfloat[$E_6$]{ \noindent\scalebox{0.55}{
        \begin{tikzpicture}[scale=0.2] 
            \tikzstyle{every node}+=[inner sep=0pt]
            \draw [black] (10.9,-13.8) circle (3);
            \draw (10.9,-13.8) node {$0$};
            \draw [black] (10.9,-13.8) circle (2.4);
            \draw [black] (28.6,-13.8) circle (3);
            \draw (28.6,-13.8) node {$1$};
            \draw [black] (13.295,-12.005) arc (120.15825:59.84175:12.848);
            \fill [black] (26.2,-12.01) -- (25.76,-11.17) -- (25.26,-12.04);
            \draw (19.75,-9.77) node [above] {$38$};
            \draw [black] (26.105,-15.455) arc (-62.66259:-117.33741:13.838);
            \fill [black] (13.4,-15.45) -- (13.88,-16.27) -- (14.34,-15.38);
            \draw (19.75,-17.5) node [below] {$63,r_{B6}$};
            \draw [black] (8.404,-12.157) arc (264.38332:-23.61668:2.25);
            \draw (6.35,-7.48) node [above] {$r_{B6}$};
            \fill [black] (10.69,-10.82) -- (11.26,-10.07) -- (10.27,-9.97);
            \draw [black] (3.8,-13.8) -- (7.9,-13.8);
            \fill [black] (7.9,-13.8) -- (7.1,-13.3) -- (7.1,-14.3);
        \end{tikzpicture} }
    }\\
    \subfloat[$E_3$]{ \noindent\scalebox{0.55}{
        \begin{tikzpicture}[scale=0.2] 
            \tikzstyle{every node}+=[inner sep=0pt]
            \draw [black] (10.9,-13.8) circle (3);
            \draw (10.9,-13.8) node {$0$};
            \draw [black] (10.9,-13.8) circle (2.4);
            \draw [black] (28.6,-13.8) circle (3);
            \draw (28.6,-13.8) node {$1$};
            \draw [black] (10.9,-28.7) circle (3);
            \draw (10.9,-28.7) node {$2$};
            \draw [black] (13.498,-12.309) arc (114.21463:65.78537:15.244);
            \fill [black] (26,-12.31) -- (25.48,-11.52) -- (25.07,-12.44);
            \draw (19.75,-10.47) node [above] {$32$};
            \draw [black] (25.847,-14.984) arc (-71.25615:-108.74385:18.974);
            \fill [black] (13.65,-14.98) -- (14.25,-15.71) -- (14.57,-14.77);
            \draw (19.75,-16.49) node [below] {$41,r_{B3}$};
            \draw [black] (8.404,-12.157) arc (264.38332:-23.61668:2.25);
            \draw (5.88,-7.48) node [above] {$r_{B3}$};
            \fill [black] (10.69,-10.82) -- (11.26,-10.07) -- (10.27,-9.97);
            \draw [black] (3.8,-13.8) -- (7.9,-13.8);
            \fill [black] (7.9,-13.8) -- (7.1,-13.3) -- (7.1,-14.3);
            \draw [black] (8.941,-26.444) arc (-147.83654:-212.16346:9.757);
            \fill [black] (8.94,-26.44) -- (8.94,-25.5) -- (8.09,-26.03);
            \draw (6.94,-21.25) node [left] {$42$};
            \draw [black] (12.31,-16.441) arc (21.53938:-21.53938:13.1);
            \fill [black] (12.31,-16.44) -- (12.14,-17.37) -- (13.07,-17);
            \draw (13.72,-21.25) node [right] {$35,r_{B3}$};
        \end{tikzpicture} }
    }
    \subfloat[$E_7$]{ \noindent\scalebox{0.55}{
        \begin{tikzpicture}[scale=0.2] 
            \tikzstyle{every node}+=[inner sep=0pt]
            \draw [black] (10.9,-13.8) circle (3);
            \draw (10.9,-13.8) node {$0$};
            \draw [black] (10.9,-13.8) circle (2.4);
            \draw [black] (28.6,-13.8) circle (3);
            \draw (28.6,-13.8) node {$1$};
            \draw [black] (10.9,-28.7) circle (3);
            \draw (10.9,-28.7) node {$2$};
            \draw [black] (13.498,-12.309) arc (114.21463:65.78537:15.244);
            \fill [black] (26,-12.31) -- (25.48,-11.52) -- (25.07,-12.44);
            \draw (19.75,-10.47) node [above] {$30$};
            \draw [black] (25.847,-14.984) arc (-71.25615:-108.74385:18.974);
            \fill [black] (13.65,-14.98) -- (14.25,-15.71) -- (14.57,-14.77);
            \draw (19.75,-16.49) node [below] {$71,r_{B7}$};
            \draw [black] (8.404,-12.157) arc (264.38332:-23.61668:2.25);
            \draw (5.88,-7.48) node [above] {$r_{B7}$};
            \fill [black] (10.69,-10.82) -- (11.26,-10.07) -- (10.27,-9.97);
            \draw [black] (3.8,-13.8) -- (7.9,-13.8);
            \fill [black] (7.9,-13.8) -- (7.1,-13.3) -- (7.1,-14.3);
            \draw [black] (8.941,-26.444) arc (-147.83654:-212.16346:9.757);
            \fill [black] (8.94,-26.44) -- (8.94,-25.5) -- (8.09,-26.03);
            \draw (6.94,-21.25) node [left] {$74$};
            \draw [black] (12.31,-16.441) arc (21.53938:-21.53938:13.1);
            \fill [black] (12.31,-16.44) -- (12.14,-17.37) -- (13.07,-17);
            \draw (13.72,-21.25) node [right] {$65,r_{B7}$};
        \end{tikzpicture} }
    }
    \subfloat[$E_8$]{ \noindent\scalebox{0.55}{
        \begin{tikzpicture}[scale=0.2] 
            \tikzstyle{every node}+=[inner sep=0pt]
            \draw [black] (10.9,-13.8) circle (3);
            \draw (10.9,-13.8) node {$0$};
            \draw [black] (10.9,-13.8) circle (2.4);
            \draw [black] (28.6,-13.8) circle (3);
            \draw (28.6,-13.8) node {$1$};
            \draw [black] (10.9,-28.7) circle (3);
            \draw (10.9,-28.7) node {$2$};
            \draw [black] (13.498,-12.309) arc (114.21463:65.78537:15.244);
            \fill [black] (26,-12.31) -- (25.48,-11.52) -- (25.07,-12.44);
            \draw (19.75,-10.47) node [above] {$72$};
            \draw [black] (25.847,-14.984) arc (-71.25615:-108.74385:18.974);
            \fill [black] (13.65,-14.98) -- (14.25,-15.71) -- (14.57,-14.77);
            \draw (19.75,-16.49) node [below] {$81,r_{B8}$};
            \draw [black] (8.404,-12.157) arc (264.38332:-23.61668:2.25);
            \draw (5.88,-7.48) node [above] {$r_{B8}$};
            \fill [black] (10.69,-10.82) -- (11.26,-10.07) -- (10.27,-9.97);
            \draw [black] (3.8,-13.8) -- (7.9,-13.8);
            \fill [black] (7.9,-13.8) -- (7.1,-13.3) -- (7.1,-14.3);
            \draw [black] (8.941,-26.444) arc (-147.83654:-212.16346:9.757);
            \fill [black] (8.94,-26.44) -- (8.94,-25.5) -- (8.09,-26.03);
            \draw (6.94,-21.25) node [left] {$82$};
            \draw [black] (12.31,-16.441) arc (21.53938:-21.53938:13.1);
            \fill [black] (12.31,-16.44) -- (12.14,-17.37) -- (13.07,-17);
            \draw (13.72,-21.25) node [right] {$73,r_{B8}$};
        \end{tikzpicture} }
    }
    \subfloat[$E_4$]{ \noindent\scalebox{0.55}{
        \begin{tikzpicture}[scale=0.2] 
            \tikzstyle{every node}+=[inner sep=0pt]
            \draw [black] (36.4,-14.6) circle (3);
            \draw (36.4,-14.6) node {$0$};
            \draw [black] (36.4,-14.6) circle (2.4);
            \draw [black] (50.3,-27.7) circle (3);
            \draw (50.3,-27.7) node {$3$};
            \draw [black] (36.4,-27.7) circle (3);
            \draw (36.4,-27.7) node {$2$};
            \draw [black] (22,-27.7) circle (3);
            \draw (22,-27.7) node {$1$};
            \draw [black] (36.026,-11.635) arc (214.9249:-73.0751:2.25);
            \draw (41.74,-7.61) node [above] {$r_{B4}$};
            \fill [black] (38.53,-12.5) -- (39.47,-12.45) -- (38.9,-11.63);
            \draw [black] (31.9,-10.1) -- (34.28,-12.48);
            \fill [black] (34.28,-12.48) -- (34.07,-11.56) -- (33.36,-12.27);
            \draw [black] (38.99,-16.112) arc (57.33085:36.06346:35.919);
            \fill [black] (48.64,-25.2) -- (48.57,-24.26) -- (47.76,-24.85);
            \draw (45.75,-19.73) node [above] {$34$};
            \draw [black] (39.337,-14.04) arc (92.85124:0.54307:10.811);
            \fill [black] (39.34,-14.04) -- (40.16,-14.5) -- (40.11,-13.5);
            \draw (53.77,-16.49) node [above] {$51,53,r_{B4}$};
            \draw [black] (23.673,-25.211) arc (143.6119:120.97507:34.755);
            \fill [black] (23.67,-25.21) -- (24.55,-24.86) -- (23.75,-24.27);
            \draw (26.75,-19.63) node [above] {$54$};
            \draw [black] (37.574,-17.354) arc (16.61227:-16.61227:13.278);
            \fill [black] (37.57,-24.95) -- (38.28,-24.32) -- (37.32,-24.04);
            \draw (38.63,-21.15) node [right] {$52$};
            \draw [black] (35.101,-25.004) arc (-161.38487:-198.61513:12.075);
            \fill [black] (35.1,-17.3) -- (34.37,-17.89) -- (35.32,-18.21);
            \draw (33.97,-21.15) node [left] {$37,r_{B4}$};
            \draw [black] (21.344,-24.783) arc (-175.5504:-279.86263:10.455);
            \fill [black] (33.56,-13.67) -- (32.86,-13.04) -- (32.68,-14.03);
            \draw (19.76,-15.75) node [above] {$39,r_{B4}$};
        \end{tikzpicture} }
    }
    \caption{Specification of the Flexible Manufacturing System}
    \label{fig:SFM_espec}
\end{figure*}
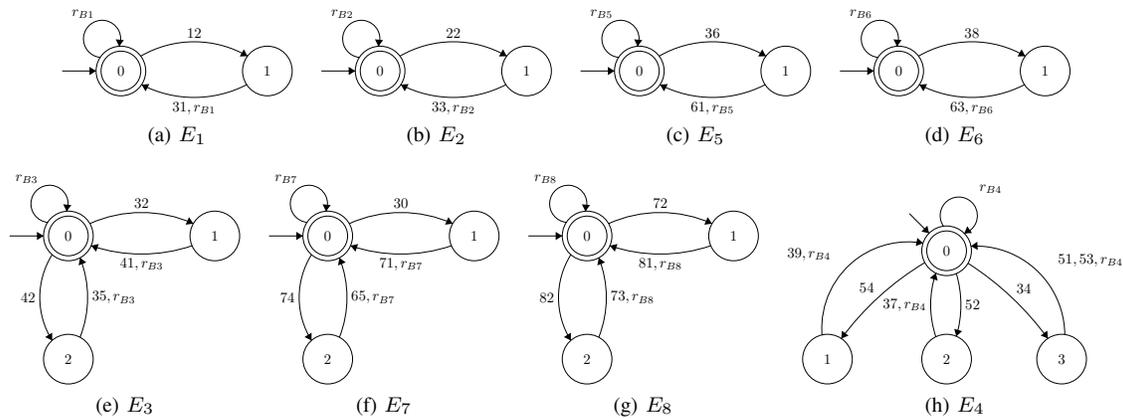


\begin{table}[h]

\caption{Supervisors of the Flexible Manufacturing System ($S_j$ relates to $E_j$), $|w|$ is the size of one of the shortest synchronizing words.}
\centering
\begin{tabular}{|c|c|c|c|c|c|}
\hline
\textbf{Sup.} & \textbf{Plants}  & \textbf{States} & \textbf{Trans.} & \textbf{Trans. $\Sigma_r$} & \textbf{$|w|$}\\ \hline
$S_1$               & $C_1,Robot$          & 18              & 94                   & 36   & 3                    \\ \hline
$S_2$               & $C_2,Robot$          & 18              & 94                   & 54   & 3                    \\ \hline
$S_3$               & $Mill, Robot$        & 18              & 90                   & 54   & 3                    \\ \hline
$S_4$               & $Lathe, Robot$       & 21              & 105                  & 63   & 3                    \\ \hline
$S_5$               & $Robot, AM$          & 44              & 253                  & 132  & 3                    \\ \hline
$S_6$               & $Robot, AM$          & 44              & 253                  & 132  & 3                    \\ \hline
$S_{7, 8}$          & $Robot, AM,$         & 260             & 2441                 & 1560 & 6                    \\   
                    & $C_3, PD$            &                 &                      &      &                      \\ \hline
\end{tabular}
\label{tab:sups}
\end{table}

Each supervisor has its own synchronizing words that allow to recover the whole system applying a partial reset. A consequence of the partial reset is that the closed loop behavior after the recovery is not led to the global initial state, but to an intermediate state where the resetted subsystems are in the initial states while the rest of the subsystems are kept untouched.

If a failure happens in one subsystem, the Mill for instance, the recovery in the two approaches will lead to different situations. The execution of a monolithic synchronizing word will take the system to the global initial state. If the Local Modular Supervisory Control is used, only the synchronizing word of supervisor $S_3$ has to be run and only the components (Robot and Mill) are going to be reinitialized. Since the robot is part of other supervisors, transitions are going to be executed in each one of the supervisors, in order to resynchronize with the new state of the robot. The states of the other subsystems will be kept the same.

\section{Conclusions}\label{sec:conclusions}

This paper presents a secure recovery procedure based on concepts of synchronizing automata and Supervisory Control Theory. This approach can be used to restore systems damaged by external attacks or temporary unobservability of some events. 

We show under what conditions the synchronicity of the plants and specifications is inherited by the composed system and supervisor and expand these results to the Local Modular Supervisory Control. Then, We present a simple modification applied to the classical modeling of systems and specifications, to include recovery events, in order to turn a regular automaton into a synchronizing one. While the monolithic approach will lead to a complete reset of the system, the application of the techniques together with the Local Modular Supervisory Control allows partial recovery of the system, resetting only the local plants and supervisors affected by the desynchronization. 

Our next steps are to adapt the recovery procedure, allowing partial resets, even in the monolithic approach; define reset procedures that do not necessarily lead to the initial state; and apply the presented recovery techniques to systems that are already inherently synchronizing.

\balance
\bibliographystyle{IEEEtran}
\bibliography{biblio}

\end{document}